\documentclass[11pt]{article}

\usepackage{amssymb,amsmath,amsthm,amsfonts}

\usepackage[T1]{fontenc}
\usepackage[tt=false, type1=true]{libertine}
\usepackage[varqu]{zi4}
\usepackage[libertine]{newtxmath}
\usepackage[margin=1in]{geometry}
\makeatletter
\def\input@path{{theme/}}
\makeatother

\usepackage{enumitem}
\setlist{nosep,topsep=0pt,leftmargin=*}

\usepackage{hyperref}
\hypersetup{
    colorlinks,
    allcolors=teal
}

\usepackage[sortcites,sorting=nyt,style=alphabetic]{biblatex}
\usepackage[ruled]{algorithm2e} 
    
    \SetAlFnt{\small}
    \SetAlCapFnt{\small}
    \SetAlCapNameFnt{\small}
    \SetAlCapHSkip{0pt} 
    \IncMargin{-\parindent}
    
\usepackage{tikz,xifthen,physics,xcolor,suffix,thm-restate,nicefrac,enumitem}
\usetikzlibrary{shapes, arrows, positioning}
\usepackage[bb=dsserif]{mathalpha}
\usepackage[capitalize]{cleveref}
\crefformat{algorithm}{Primal Budget Pacing Algorithm}
   
\definecolor{darkgreen}{rgb}{0,0.5,0}

\usepackage[suppress]{color-edits}
\definecolor{@gray}{HTML}{edc3c5}
\addauthor[Eva]{et}{brown}
\addauthor[Yishay]{ym}{darkgreen}
\addauthor[Giannis]{gf}{cyan}
\addauthor[Yoav]{yk}{olive}
\addauthor[Bobby]{rdk}{blue}
\addauthor[Gemini]{ai}{teal}
\addauthor[\textbf{TODO}]{todo}{red}

\renewcommand{\paragraph}[1]{\smallskip\noindent\textbf{#1.}}
\newcommand{\itparagraph}[1]{\smallskip\noindent\textit{#1.}}

\newcommand{\orderT}[1]{\tilde{\mathcal{O}}\qty(#1)}
\WithSuffix\newcommand\orderT*[1]{\tilde{\mathcal O}(#1)}
\DeclareMathOperator*{\argmax}{arg\,max}

\newcommand{\floor}[1]{\left\lfloor{#1}\right\rfloor}
\WithSuffix\newcommand\floor*[1]{\lfloor{#1}\rfloor}
\newcommand{\ceil}[1]{\left\lceil{#1}\right\rceil}
\WithSuffix\newcommand\ceil*[1]{\lceil{#1}\rceil}

\let\e\varepsilon

\let\d\delta

\newcommand{\opt}{\mathtt{OPT}}

\renewcommand{\b}[2]{b_{#2}\ifthenelse{\isempty{#1}}{}{^{(#1)}}}
\newcommand{\vecb}[2]{\vec{b}_{#2}\ifthenelse{\isempty{#1}}{}{^{(#1)}}}
\newcommand{\x}[2]{x_{#2}^{(#1)}}
\newcommand{\bmin}{b_{\min}}
\newcommand{\bmax}{b_{\max}}
\newcommand{\bavg}{b_{\textrm{avg}}}
\newcommand{\davg}{\delta_{\textrm{avg}}}
\newcommand{\rhomin}{\rho_{\min}}

\newcommand{\rL}{\rho_{\mathtt{L}}}
\newcommand{\rO}{\rho_{\mathtt{O}}}

\newcommand{\Line}[4]{%
    #1&%
    \ifthenelse{\isempty{#2}}{\phantom{=}}{#2}%
    #3%
    \ifthenelse{\isempty{#4}}{}{&&\qquad\left(\big.\substack{#4}\right)}%
}

\newcommand{\ignore}[1]{}
\newcommand{\Description}[1]{}
\newtheorem{theorem}{Theorem}[section]
\newtheorem*{theorem*}{Theorem} 

\newtheorem*{conjecture*}{Conjecture} 
\newtheorem{lemma}[theorem]{Lemma}
\newtheorem{corollary}[theorem]{Corollary}
\newtheorem{proposition}[theorem]{Proposition}

\newtheorem*{claim*}{Claim} 

\newtheorem*{remark}{Remark}
\theoremstyle{definition}
\newtheorem*{definition*}{Definition}

\newtheorem*{example*}{Example}

\newenvironment{myproof}[1][]{%
    \begin{proof}[#1]%
}{%
    \end{proof}%
}

\newcommand{\citet}[1]{\textcite{#1}}
\newcommand{\citep}[1]{\parencite{#1}}

\title{Robust Temporal Guarantees in Budgeted Sequential Auctions}

\author{
Giannis Fikioris%
\thanks{Supported by the Google PhD Fellowship and the ONR MURI grant N000142412742.}\\
Cornell University\\
\texttt{gfikioris@cs.cornell.edu}
\and
Robert Kleinberg%
\thanks{Supported by NSF grant CCF-2402851.}
\\
Cornell University\\
\texttt{rdk@cs.cornell.edu}
\and
Yoav Kolumbus\\
Cornell University\\
\texttt{yoav.kolumbus@cornell.edu}
\and
Yishay Mansour%
\thanks{Partially supported by the European Research Council (ERC) under the European Union's Horizon 2020 research and innovation program (grant agreement No. 882396), by the Israel Science Foundation, the Yandex Initiative for Machine Learning at Tel Aviv University, and a grant from the Tel Aviv University Center for AI and Data Science (TAD).}\\
Tel Aviv University and Google Research\\
\texttt{mansour.yishay@gmail.com}
\and
\'Eva Tardos%
\thanks{Supported in part by AFOSR grant FA9550-23-1-0410, AFOSR grant FA9550-231-0068, ONR MURI grant N000142412742, and through a SPROUT Award from Cornell Engineering.}\\
Cornell University\\
\texttt{eva.tardos@cornell.edu}
}

\date{\vspace{-25pt}}

\begin{document}


\maketitle{}
\thispagestyle{empty}

\begin{abstract}
    In modern advertising platforms, learning algorithms are deployed by budget-constrained bidders to maximize their accumulated value. These algorithms often offer classical utility guarantees like no-regret, i.e., the agent's utility is at least the utility achieved by some benchmark in which it is assumed that every other agent's bidding remains the same. These guarantees offer compelling properties: They are optimal against stationary competition distributions, and in unconstrained settings, the resulting empirical distribution of play induced by no-regret dynamics approximates a Coarse Correlated Equilibrium. However, no-regret algorithms are easily manipulable, and in budgeted settings, no stronger notion of regret (such as swap regret) is currently known that would limit such manipulation.

We propose a very simple learning algorithm for budgeted sequential auctions where agents maximize their total number of wins and show that it has surprisingly appealing properties. We analyze this algorithm from two perspectives. First, we show that when an agent with a $\rho$ fraction of the total budget uses this algorithm, then she is guaranteed to win at least $\rho T - O(\sqrt T)$ of the total $T$ rounds. This result holds for adversarial behavior by the other agents, as long as they respect their own budget restrictions. Second, we examine the scenario when all the agents follow our algorithm. By the first result, every agent's total wins are proportional to her budget, up to the additive $O(\sqrt T)$ term. In addition, we show that this result holds in a much stronger sense: after an initial period of $O(\sqrt T \log T)$ rounds, every agent gets the same guarantee over any time interval. For intervals of length $O(\sqrt T)$, we show that the deviation from the desired number of wins is an additive constant.
\end{abstract}

\clearpage
\setcounter{page}{1}
\section{Introduction}

Advertising auctions form a trillion-dollar industry, and participants in this ecosystem all use learning algorithms to set their bids. There is a rich theory of learning outcomes in games \cite{cesa2006prediction,DBLP:journals/corr/abs-1909-05207}; however, the classical theory of learning in games does not handle agents with global budget constraints. In settings without budget constraints, learning algorithms can achieve appealing properties like no-regret guarantees, and when all agents in a game achieve no-regret, the distribution of outcomes forms a Coarse Correlated Equilibrium of the game \cite{Hannan57,hart2000simple}. 
In particular, also in first-price auctions, which have become the predominant format in online advertising markets, without budget constrains no-regret dynamics would approach the set of Coarse Correlated Equilibria.

However, participants in modern advertising platforms are budget-constrained, and much less is known about how to learn to bid well in auctions while respecting a budget limit. 
In particular, classical no-regret algorithms cannot deal with budget constraints. 
While there are algorithms that offer no-regret guarantees against stationary environments while respecting a budget constraint \cite{DBLP:journals/mansci/BalseiroG19,DBLP:journals/corr/AggarwalFZ24,DBLP:conf/colt/LucierPSZ24}, their guarantees against adversarial environments are much weaker (e.g.,  \cite{DBLP:journals/jacm/BadanidiyuruKS18}). Further, nothing is known about the outcome
when a no-regret agent plays against an optimizer, and what is known about the outcomes of self-play
among such algorithms is very limited, as we discuss further in \Cref{sec:related}.

One important shortcoming of no-regret learning algorithms, even in the unbudgeted setting, 
is that they are manipulable: a player who seeks to extract higher payoffs against learners (also called an optimizer), e.g., one of the competing bidders, or the platform itself setting reserve prices, can manipulate others who are using classical no-regret algorithms \cite{braverman2018selling}, steering the learning dynamics in the optimizer's favor. 
This implies that no-regret learning is not a best response in repeated games when others are using no-regret learning algorithms. More sophisticated learning algorithms can guarantee no-swap-regret, which shields the agents against some form of manipulation: an optimizer can behave as a Stackelberg leader, but cannot do better than that against a no-swap-regret learner, see \cite{DBLP:conf/nips/DengSS19, DBLP:conf/colt/MansourMSS22, DBLP:conf/sigecom/Arunachaleswaran25, DBLP:conf/sigecom/Arunachaleswaran24}. 
However, no analogue of no-swap-regret is known in the budgeted setting, even when allowing for multiplicative guarantees.


Finally, advertisers may care not only about how often they win, but also about how their winnings are distributed over time. Recent work by \citet{DBLP:conf/sigecom/FikiorisKKKMT25} studies a model of such temporal spacing preferences for the bidders, and offers a complicated learning algorithm to achieve temporal spacing guarantees in a static environment, but not much is known about the spacing properties of other algorithms, or about temporal spacing in non-stationary settings, such as games.

\subsection{Our Results}

In this paper, we consider a simple learning algorithm that updates bids $\b{t}{}$ depending on the price paid in the previous iteration, using a very simple and deterministic update function and a single parameter $\eta$.
Specifically, the update rule is
\begin{equation}
    \b{t+1}{}= \b{t}{}
    +
    \eta\qty(\rho_i - p_i^{(t)}),
\end{equation}
where $\b{t}{}$ is the bid in iteration $t$, $\rho_i T$ is the bidder's total budget over all $T$ rounds, 
$\eta$ is the learning rate, 
and $p_i^{(t)}$ is the amount the bidder paid in the previous iteration (which is $0$ if they didn't win for first/second-price auctions). 
In the introduction, we state the results using $\eta =1/\sqrt{T}$,
which yields the strongest results quantitatively. 


We will consider this simple learning algorithm in first-price auctions, and will show that, despite its extreme simplicity, it does very well with respect to all the issues discussed above in environments where budgeted bidders maximize their total number of wins, that is, when all wins are viewed as equally valuable.

\paragraph{Algorithm against an Optimizer}
We show that our proposed algorithm is robust against arbitrary bidding behavior by other agents in both first- and second-price auctions. 
Specifically, we show that when an agent $i$ has budget $\rho_i T$ while all other agents together have budget $(1 - \rho_i) T$, if agent $i$ uses our algorithm, she is guaranteed to win at least $\rho_i T - \order*{\sqrt T}$ times, no matter how the other agents bid, as long as their budget limit is enforced.


\paragraph{Discrepancy of wins}
We also show that when all players use this simple learning algorithm in repeated first-price auctions, the wins of each agent are approximately evenly distributed over time. 
Our previous result shows that every agent $i$ with budget $\rho_i T$ is guaranteed to win roughly $\rho_i T$ rounds (assuming $\sum_i \rho_i = 1$).
With a $\rho_i$ fraction of the wins overall, an ideal (but often infeasible) distribution of wins would have agent $i$ win $\rho_i \tau$ times in every period of length $\tau$.
If the number of times agent $i$ wins in period of length $\tau$ differs from the ideal
$\rho_i \tau$ wins, we refer to the difference
as the \emph{discrepancy} of agent $i$'s wins in this period. 
Our main result is to show that after an initial period of $\order*{\sqrt T \log T}$ rounds, the wins have very small discrepancy.
Specifically, for every window of size\footnote{For simplicity of presentation, in this section we suppress dependence of any other parameters in our $\order*{\cdot}$ bounds.} $\tau = \order*{\sqrt T}$ we show that every agent $i$ gets $\rho_i \tau \pm \order{1}$ wins.
That is, each agent $i$ has at most $\order{1}$ discrepancy in every period of length $\tau =\order*{\sqrt T}$ (after the initial startup stage).

\itparagraph{Discrepancy of wins for equal budgets}
In the case of $n$ agents with equal budgets (that is $\rho_i = 1/n$ for all agents $i$), we show that after an initial period of $\order*{\sqrt T}$ rounds, the $n$ agents will simply win in a round-robin order, implying a discrepancy less than $1$ (specifically at most $\frac{n-1}{n}$) for \emph{any agent and any period}.\footnote{In the case with general budgets, having every agent win in exactly every $1/\rho_i$ rounds is not possible even with a central coordinator (which is also called perfect periodic schedules), for our setting consider  budget shares of $\frac12, \frac13, \frac16$ for a simple example.}

\subsection{Our Techniques}

Note that the update rule adjusts bids based on the price that the bidder paid in the last step.
The first thing we show is that this update rule guarantees that the bidder cannot run out of budget, unless they start with too high a starting bid (\cref{lef:model:budget_satisfaction}). 

To prove the overall non-manipulability of the algorithm, we notice that with the deterministic updating of the learner, the problem of the optimizer can be phrased as an integer program \cref{eq:optimization}, where the optimizer chooses which rounds she wants to win.
We prove the claimed upper bound on how much the optimizer can win by considering a Lagrangian relaxation of the linear program.
Specifically, we upper bound the Lagrangian by two additive terms: one linear in the number of rounds remaining (bounding how much value the Optimizer can get in the long run) and one depending on the Learner's current bid (bounding how much extra benefit there is from the current state of the learner's algorithm).
This result shows that, despite the deterministic behavior of our algorithm (which suggests it is more susceptible to manipulation than a randomized algorithm), the way we update the bids is exactly right: the optimizer has no benefit of either letting the bid get too low, winning at lower prices, or letting the bid get too high, potentially wasting the learner's budget.

The more technically challenging result of our paper is bounding the discrepancy of wins in all sub-intervals. 
The crux of the proof lies in showing that when one scales the units of money (the numeraire) so that budgets sum up to $T$, then after a short initial interval all bids will converge to an interval of width $\order{\eta}$ surrounding $1$, where $\eta$ is the learning rate.
Once all the bids are inside this interval, we show that they will remain inside the interval, and that this implies the desired bound on the discrepancy (\cref{lem:diff_b:wins_deviation}).

To show the fast convergence to all bids equaling almost $1$, we first show that the proposed learning dynamics corresponds to subgradient descent on a convex non-smooth function of the bids; see~\cref{eq:func_gradient}.
The unique optimum of this function is when all bids are $1$, which is a natural equilibrium of the bidding game with the $\sum_i \rho_i = 1$ assumption.
Using the \textit{quadratic growth} property\footnote{A function $f$ is said to satisfy quadratic growth if $f(x) - f(x^\star) \ge \Omega(\norm{x - x^\star}_2^2)$, where $x^\star$ is the minimizer of $f$.} of our subgradient-descent objective, this allows us to show that after $\order*{\frac{1}{\eta} \log\frac{1}{\eta}}$ rounds, all bids will be within $\order*{\sqrt\eta}$ of $1$ (\cref{ssec:diff_b:sqrt_eta}).
This implies a guarantee of $\order*{1/\sqrt\eta}$ discrepancy for small enough interval lengths $\tau$.

To show the claimed bound of $\order*{1}$ discrepancy, we need to improve the bound on the difference between the bids and $1$, from $\order{\sqrt{\eta}}$ to $\order{\eta}$.
This would follow from existing work like \cite{DBLP:journals/mor/CharisopoulosD24}, if our function would satisfy the \textit{sharpness} property.\footnote{A function $f$ is said to be sharp if $f(x) - f(x^\star) \ge \Omega(\norm{x - x^\star}_2)$, where $x^\star$ is the minimizer of $f$.} However, our function does \textbf{not} satisfy this property, and the bids' distance from $1$ can increase between rounds if it is $\order{\sqrt\eta}$.
We prove the desired convergence in a two-step process. First, we argue that after an additional $\order*{1/\eta}$ rounds, all the agents' bids are within $\order*{\eta}$ of the average bid (\cref{ssec:diff_b:eta_of_avg}).
Once all bids are so close to each other, all we have to show is that the average bid will be within $\order{\eta}$ of $1$. We show this happens after an additional $\order*{\frac{1}{\eta} \log\frac{1}{\eta}}$ rounds (\cref{ssec:diff_b:eta_of_one}).

\subsection{Related Work}
\label{sec:related}

\textbf{Learning in Auctions with Budgets.}
Learning in auctions has been an active topic of study for decades, but the introduction
of budget constraints into models of learning in auctions came under investigation only
recently. The seminal work of \citet{DBLP:journals/mansci/BalseiroG19} 
proposes an algorithm known as \emph{Adaptive Pacing} that shades bids
by a factor known as a \emph{pacing multiplier} and adjusts pacing
multipliers using an operation that can be interpreted as gradient descent 
in a Lagrangian dual space. Their work shows that adaptive pacing achieves
an optimal competitive ratio against worst-case competition, sublinear regret 
against stationary competition, and approximate Nash equilibrium in self-play
under a ``large market'' assumption.
Subsequently \citet{DBLP:conf/innovations/GaitondeLLLS23} 
proved that adaptive pacing algorithms in repeated auctions 
achieve a 2-approximation to the optimum \emph{liquid welfare}, 
a measure of social welfare adapted for budgeted settings;
this result was extended by \citet{DBLP:conf/colt/LucierPSZ24} 
to incorporate ROI constraints in addition to budget constraints. 
In the specific context of first-price auctions, \cite{DBLP:conf/icml/WangYDK23} and \cite{DBLP:journals/corr/AggarwalFZ24} extend these learning results, targeting no-regret against stationary environments.
An even broader family of bidding algorithms --- those that have bounded 
regret with respect to the best pacing multiplier in hindsight ---
were shown by \citet{fikioris2025mor} to guarantee a constant approximation
to the optimum liquid welfare.
We refer the reader to \cite{aggarwal2024auto} for a survey covering a broader set of results on this area.

Our work differs from the budget pacing work by \citet{DBLP:journals/mansci/BalseiroG19}
and its successors by focusing on a simple primal update rule, rather than dual-variable pacing.
Our results show that this simple rule satisfies robust guarantees in adversarial settings
and, when used by all agents, leads to an approximately budget-proportional allocation, even over 
short subintervals of the timeline.


\textbf{Pacing Dynamics and Primal Updates.}
Most budgeted learning algorithms rely on ``dual pacing,'' where bids are scaled down by a multiplier that is updated over time.  \citet{DBLP:journals/corr/abs-2202-06152} demonstrated that certain pacing behaviors can be modeled as Convolutional Mirror Descent (CMD), providing a connection between primal and dual dynamics. Our algorithm employs a considerably simpler, purely primal update rule—adjusting bids directly based on payment versus target spend—which allows us to prove strong temporal properties that are difficult to analyze under complex dual dynamics.


\textbf{Robustness and Manipulability.}
A major critique of standard no-regret learning is its manipulability. 
\citet{braverman2018selling} showed that a strategic optimizer can exploit 
no-regret learners to extract maximum revenue when selling an item.
Their work stimulated a long series of papers characterizing 
the potential outcomes achievable when a learning 
algorithm interacts with an optimizer that knows the learner's 
algorithm and may manipulate it to achieve the optimizer's goals
(see, e.g., \cite{cai2023selling,Guruganesh2024,lin2024generalized}). 
In unbudgeted settings, the optimizer's advantage 
can be mitigated by algorithms minimizing \emph{swap regret}, 
which limits the optimizer's gain to a Stackelberg leader's advantage 
\cite{DBLP:conf/nips/DengSS19, DBLP:conf/colt/MansourMSS22, DBLP:conf/sigecom/Arunachaleswaran25,rubinstein2024strategizing}. 
In the case of auctions with stochastic values for the buyer, the Stackelberg leader value is equivalent
to Myerson's optimal revenue. 
\citet{DBLP:conf/sigecom/KumarSS24} develop learning algorithms that are simpler than those achieving no-swap-regret, yet guarantee that the buyer does not have to pay more than Myerson's optimal revenue.
However, no analogue of swap regret is known for budgeted settings.

We take a different approach to robustness.
Instead of seeking an equilibrium concept like swap regret, we provide a ``safety level'' guarantee: an agent using our algorithm secures a proportional share of wins regardless of the opponent's strategy, effectively shielding them from any strategies employed by an optimizer.

\textbf{Discrepancy and Non-Manipulability in Self-play.}
There is not much literature about self-play by learning algorithms in budgeted settings, i.e.,~properties of the outcome when all players use one of the proposed algorithms.
\citet{DBLP:journals/mansci/BalseiroG19} is the only paper that discusses this for repeated auctions with agents who have quasi-linear utilities.
Their value setting is more general than ours, as agents have random values about how much they care for winning a round.
They discuss two results that are related to ours.
First, when the agents perform self-play and use their proposed learning algorithm, the outcome of play converges to a Second Price Pacing Equilibrium (\cite{DBLP:conf/wine/ConitzerKSM18}).
However, their results only ensure that bids get within $\Theta(\sqrt \eta)$ of the desired equilibrium, which would only prove an $\order*{1/\sqrt\eta}$ discrepancy bound for us.
Second, they prove that when everyone is using their algorithm, small bidders in large markets are approximately best-responding by using their learning algorithm, i.e., when all players use their learning algorithm, a single small player cannot significantly manipulate the outcome in their own favor. In contrast, we show non-manipulability with no assumptions.

\textbf{Temporal Spacing and Win Distribution.}
Finally, our results on the distribution of wins relate to the emerging study of ``spacing'' in 
repeated auctions. In the marketing literature it is well documented that 
the spacing of users' exposure to advertising over time is related to the ads' effectiveness
\citep{broadbent1993advertising,broadbent1995adstock,broadbent2000advertisements,craig1976advertising,ha1997does,weinberg1982econometric}.
\citet{DBLP:conf/sigecom/FikiorisKKKMT25} recently introduced 
a model of the advertisers' objectives that accounts for
the advantage of evenly spacing wins over time, and they presented 
algorithms for optimizing spacing of wins under budget constraints
when bidding in repeated second-price auctions in stationary environments.
Rather than focusing on stationary environments, our work focuses
on the dynamics resulting from multiple bidders simultaneously running simple
learning algorithms, showing that such dynamics naturally lead to low-discrepancy 
(well-spaced) outcomes without requiring the complex explicit spacing objectives 
used in prior work.

\section{Definitions and Learning Algorithm}

We consider sequential first-price auctions over a time period of $T$ rounds.
There are $n$ agents, each trying to maximize their total number of wins.
Each agent $i$ has a budget of $\rho_i T$ that she uses to bid and has no other benefit if any budget remains.
We normalize the budgets so that $\sum \rho_i = 1$.

\subsection{Learning Algorithm}

In this section, we propose the learning algorithm that the agents are using and prove some basic facts about it.
Our algorithm simply increases/decreases the bid of the next round if the payment of the current round is less/more than the average budget per round.
We depict the full algorithm in the Figure above:  \cref{algo}.

\begin{algorithm}[t]
\DontPrintSemicolon
\caption{Primal Budget Pacing Algorithm}
\label{algo}
\KwIn{Learning rate $\eta$, initial bid $\b{1}{i} \ge 0$, budget share $\rho_i$}

\For{$t \in [T]$}
{
    Bid $b^{(t)}$\;

    Observe payment $p^{(t)}$ and calculate bid of next round:
    \begin{minipage}{.9125\linewidth}
    \begin{equation}\label{eq:update}
        \b{t+1}{i}
        =
        \b{t}{i}
        +
        \eta\qty(\rho_i - p_i^{(t)})
        .
    \end{equation}
    \end{minipage}
}
\end{algorithm}

As we show in the subsequent sections, this algorithm offers powerful properties, both against arbitrary bidding by other agents and when deployed by multiple players.
Aside from these strong properties, this \cref{algo} also offers desirable properties like a simple update rule (i.e., the updated bid always remains non-negative without the need for projection) and the bidder not running out of budget.

\begin{lemma}\label{lef:model:budget_satisfaction}
    Fix an agent $i$ who is running our algorithm with $\eta \in (0, 1)$ in either a first or second price auction, and let $\b{1}{i}$ be the initial bid.
    Then, for arbitrary behavior by the other agents, the algorithm's bids are always non-negative, and the agent's total payment by round $T$ is
    \begin{equation*}
        \rho_i T + \frac{\b{1}{i} - \b{T+1}{i}}{\eta}
        .
    \end{equation*}
    In addition, if $\b{1}{i} \le \rho_i$ then the agent never runs out of budget.
\end{lemma} 

\begin{myproof}
    To prove that the bids are non-negative, we notice that
    \begin{equation}\label{eq:531}
        \b{t+1}{i}
        =
        \b{t}{i}
        +
        \eta\qty(\rho_i - p_i^{(t)})
        \ge
        \b{t}{i}
        +
        \eta\qty(\rho_i - \b{t}{i})
        =
        (1-\eta)\b{t}{i}
        +
        \eta \rho_i
        \ge
        \min\qty{\b{t}{i}, \rho_i}
    \end{equation}
    where the first inequality follows from the payment rule, either first- or second-price, and the second inequality from $0 < \eta < 1$.
    Therefore, if $\b{1}{i} \ge 0$, then every bid is non-negative.

    By summing \eqref{eq:update} over all $t \in [T]$ we get that $\sum_{t \in [T]} p_i^{(t)} = \rho_i T + \frac{\b{1}{i} - \b{T+1}{i}}{\eta}$ which proves the payment part of the claim.
    We get that the agent does not run out of budget by \cref{eq:531}, which recursively proves $\b{T+1}{i} \ge \min\{ \b{1}{i}, \rho_i \} = \b{1}{i}$, if $\b{1}{i} \le \rho_i$.
\end{myproof}

\begin{remark}
    Because of $\cref{lef:model:budget_satisfaction}$, we assume that if an agent $i$ runs the \cref{algo}, their initial bid is at most $\rho_i$, thereby ensuring they do not exceed their budget.
\end{remark}

\section{Primal Budget Pacing Algorithm against budgeted opponent} 
\label{sec:adv}

For this section, we consider a single agent who is using the \cref{algo} to bid, while the rest of the agents all collude together to maximize their number of wins.
To differentiate between the two, in this section, we call the first agent \textit{Learner} with a budget of $\rL T$, and we combine the rest of the agents into an agent we call \textit{Optimizer} with a budget of $\rO T$.
The main result of this section is that the Optimizer can win at most $\frac{\rO}{\rL + \rO}T + 3\sqrt T$ times, which implies that the Learner is always guaranteed to win $\frac{\rL}{\rL + \rO} T - 3\sqrt T$ times against any opponent.

\begin{theorem}\label{thm:adv:main}
    Fix a Learner with budget $\rL T$ who is using the \cref{algo} with initial bid $\b{1}{} = \rL$ and learning rate $\eta = 1/\sqrt T$.
    Then, an Optimizer with budget $\rO T$ cannot win more than $\frac{\rO}{\rL + \rO} T + 3\sqrt T$ times.
    In addition, this result holds when the underlying auction is either first or second-price.
\end{theorem}

For the rest of the section, we provide the necessary steps that show the theorem above.
First, we notice that by our choice of the Learner's initial bid and \cref{lef:model:budget_satisfaction} we never have to worry about the Learner running out of budget.

We now simplify the Optimizer's strategy space.
For simplicity, we assume that when the Optimizer and the Learner make the same bid, the Optimizer picks the winner\footnote{All our results extend to arbitrary tie-breaking rules, by having the Optimizer bid the Learner's bid $\pm \e$ for an arbitrarily small $\e > 0$}.
Let $\b{t}{}$ be the Learner's bid in round $t$.
Note that $\b{t}{}$ is a deterministic function of the past rounds, making it known to the Optimizer.
We show that the Optimizer's bid in round $t$ should always be the Learner's bid $\b{t}{}$, making his only decision in round $t$ whether to win or not.
This guarantees a win at the lowest possible price, or losing but ensuring that the Learner pays as much as possible.
This fact about the Optimizer's bids is what makes \cref{thm:adv:main} hold in both first and second-price formats: both players' bids are the same in either payment rule.

We formalize the above statement in the following proposition.
While the statement of the proposition is complicated to make it formal, its message is simple:
Every bidding sequence by the Optimizer can be replaced by another bidding sequence that matches the Learner's bids that has the same wins and decreases the Optimizer's total payment.
The key statement that makes this true is that under this change, the Learner's bids will never increase.
Whenever the Learner wins, she pays more, making her future bids smaller.
Whenever the Learner loses, her update is a constant increase, so this does not increase her future bids either.

\begin{proposition}
    Fix a sequence of bids $\b{t}{O}$ by the Optimizer and let $\b{t}{}$ be the corresponding bid sequence by the Learner.
    Let $\x{t}{} \in \{0, 1\}$ be the indicator variable such that the Optimizer wins round $t$ if and if $\x{t}{} = 1$.
    Now let $\hat b^{(t)}_O$ be a new bidding sequence by the Optimizer and $\hat b^{(t)}$ be the induced sequence by the Learner, which are defined recursively:
    $\hat b^{(t)}_O = \hat b^{(t)}$ and all the ties are broken to ensure the same sequence of wins $\x{t}{}$.
    Then the new bidding sequence ensures the same number of wins for the Optimizer, but at a lower cost.
    This result is true when the underlying auction format is either first or second-price.
\end{proposition}


\begin{myproof}
    All we have to prove is that the Learner's new bids are not greater: for all $t$, $\hat b^{(t)} \le \b{t}{}$.
    This ensures that the Optimizer's payment in every round does not increase:
    When the Optimizer wins, ($\x{t}{} = 1$) we have
    \begin{equation*}
        \hat b^{(t)}_O
        \overset{\textrm{by definition}}{=}
        \hat b^{(t)}
        \le
        \b{t}{}
        \overset{\x{t}{} = 1}{\le}
        \b{t}{O}
    \end{equation*}
    which ensures that the Optimizer's payment does not increase for both auction formats.

    We prove that $\hat b^{(t)} \le \b{t}{}$ inductively on $t \in [T]$.
    First, for $t = 1$, we have $\hat b^{(1)} \le \b{1}{}$ since the Optimizer cannot change the Learner's starting bid.
    Now fix a $t \ge 2$.
    In round $t-1$ the Learner's payment is $(1 - \x{t-1}{}) \hat b^{(t+1)}$, making her bid in round $t$:
    \begin{equation*}
        \hat b^{(t)}
        =
        \hat b^{(t-1)} \qty(1 - (1 - \x{t-1}{})\eta)
        +
        \eta \rL
        \le
        b^{(t-1)} \qty(1 - (1 - \x{t-1}{})\eta)
        +
        \eta \rL
        \le
        b^{(t)}
    \end{equation*}
    where the first inequality holds by the induction hypothesis $\hat b^{(t-1)} \le b^{(t-1)}$ and the second inequality holds because, when the Optimizer uses bid $\b{t-1}{O}$, the Learner's payment is $(1 - \x{t-1}{}) b^{(t+1)}$ in first-price and at most this in second-price.
    This completes the proof.


\end{myproof}

By the above proposition, we can restrict the Optimizer's action space to binary sequences $\{\x{t}{}\}_{t} \in \{0, 1\}^T$.
Specificically, the Optimizer's optimization problem is 
\begin{equation}\label{eq:optimization}
\begin{split}
    \max_{\x{1}{}, \ldots, \x{T}{} \in \{0, 1\}} 
    &\quad
    \sum_{t=1}^T \x{t}{}
    \\
    \text{such that}
    &\quad
    \sum_{t=1}^T \x{t}{} \b{t}{} \le T \rO
    \\
    \text{where}
    &\quad
    \b{t+1}{}
    =
    \b{t}{} + \eta \qty(\Big.\rL - \b{t}{} \qty(1 - \x{t}{}))
    \quad \forall t
\end{split}
\end{equation}

We will show that the above is at most $\frac{\rO}{\rL + \rO}T + 3\sqrt{T}$ as in \cref{thm:adv:main}.
We do so by examining the Lagrangian objective $\lambda T \rO + \sum_{t=1}^T \x{t}{} (1 - \lambda \b{t}{})$ for a specific Lagrange multiplier $\lambda \ge 0$.
We notice that to maximize this Lagrangian, the Optimizer has to take two independent variables into account: the number of rounds and the Learner's current bid.
We show an upper bound for this Lagrangian that separates these two variables.
The upper bound is the summation of a term linear in $T$ (but independent of $b$), and a term quadratic in $b$ (but independent of $T$).
In fact, we show this for an arbitrary number of rounds $\tau$ remaining and the bid of the Learner in that round.
This helps us prove the claim inductively, and also bounds the effect of the Learner having an arbitrary bid, since the Optimizer can spend some rounds to ``make'' her have any bid.

We present the following lemma that offers the above upper bound.
For simplicity of presentation, we present a weaker version of our actual statement, where $\rL + \rO = 1$ and some other terms are simplified.
To prove the lemma, we use induction on $\tau$.
The full version of the lemma and the detailed proof can be found in \cref{lem:adv:opt_lagrange_unnormalized}.

\begin{lemma}[Weaker version of \cref{lem:adv:opt_lagrange_unnormalized}] \label{lem:adv:opt_lagrange_normalized}
    Assume $\rL + \rO = 1$ and fix $\lambda = 1 - \rO$.
    Fix an interval $[T-\tau + 1, T]$ of length $\tau$ for any $\tau \in [0, T]$ and the Learner's initial bid in that interval $b = \b{T-\tau+1}{}$.
    Then, for any binary sequence $\{\x{t}{}\}_{t \in [T-\tau+1, T]}$ it holds
    \begin{equation*}
        \lambda \tau \rO
        +
        \sum_{t=T-\tau+1}^T \x{t}{} \qty(1 - \lambda \b{t}{})
        \le
        \qty(\rO + \eta) \tau 
        +
        \qty(\frac{1}{2} b^2 - (1 + \rO) b + 2) \frac{1}{\eta}
    \end{equation*}
\end{lemma}

We now briefly explain the meaning behind the terms in the above lemma.
The $\rO \tau$ term is how many wins the Optimizer gets when always bidding $1$.
In this case, the Learner will quickly converge to also bidding about $1$.
Additionally, to remain at that bid, the Learner has to be winning a $\rL=1-\rO$ fraction of the rounds.
Under those conditions ($\b{t}{} = 1$ and $\x{t}{} = \rO$), the Optimizer's Lagrangian is exactly $\rO \tau$.

However, the Optimizer can use more complicated strategies.
On the one hand, she can perform the same strategy as above, but lose when the Learner bids $1$ and win when the Learner bids $1 - \order{\eta}$.
This gives her an excess of about $\eta\tau$ budget, which explains the corresponding term.
On the other hand, the Optimizer can make bids very far from $1$.
This is reflected in the quadratic term: if $b \le 2$, the bound increases as $b$ becomes smaller, since the Optimizer is paying less per win.
However, this quadratic is bounded by $2/\eta$ in that range, showing that there is little to gain by making the Learner use bids in the $[0, 2]$ range.
For $b > 2$, our bound is increasing in $b$, which would suggest that if the Learner started at some really high bid, the Optimizer would get really high utility.
However, we can select the Learner's bid in round $1$, and our proof of the lemma implicitly shows that the benefit of making the Learner's bid high (which requires the Optimizer to pay those bids) is not worth the additional utility.

\begin{myproof}[Proof Sketch]
    We first notice that there is a simple recursive equation for the optimal value of the Lagrangian: if $\opt_\tau(b)$ is the maximum value when there are $\tau$ rounds remaining and the Learner starts at bid $b$, then
    \begin{equation*}
        \opt_\tau(b)
        =
        \lambda \rO
        +
        \max\qty{\Big.
            1 - \lambda b + \opt_{\tau-1}(b + \eta \rL)
            ,
            \opt_{\tau-1}\qty\big(b + \eta (\rL - b))
        }
    \end{equation*}
    where the first term of the maximum corresponds to $\x{t}{} = 1$ and the second to $\x{t}{} = 0$.
    By writing our upper bound on the Lagrangian as $A \tau + g(b)$ for some $A \ge 0$ and function $g$, by induction on the above equation, what we have to prove is that
    \begin{equation*}
        A + g(b)
        \ge
        \lambda \rO
        +
        \max\qty{\Big.
            1 - \lambda b + g(b + \eta \rL)
            ,
            g\qty\big(b + \eta (\rL - b))
        }
    \end{equation*}

    By approximating $g(b + \eta x) - g(b) \approx \eta x g'(b)$ the above becomes
    \begin{equation*}
        A
        \ge
        \lambda \rO
        +
        \max\qty{\Big.
            1 - \lambda b + \eta \rL g'(b)
            ,
            \eta (\rL - b) g'(b)
        }
    \end{equation*}
    By our choice of $g(b) = \frac{b^2 / 2 - (1 + \rO) b + 2}{\eta}$, the above r.h.s. becomes exactly $\rO$, which is what we set our $A$ to, ignoring any $\eta$ terms.
\end{myproof}

We now use \cref{lem:adv:opt_lagrange_normalized} to prove \cref{thm:adv:main} for budgets normalized to $1$.
We present the proof for unnormalized budgets in \cref{sec:app:model}.

\begin{myproof}[Proof of \cref{thm:adv:main} for normalized budgets]
    We use \cref{lem:adv:opt_lagrange_normalized} for the entire time interval $[1, T]$ of length $\tau = T$ and initial bid by the Learner $b = \b{1}{} = \rL = 1 - \rO$.
    This upper bounds the Lagrangian, and therefore the Optimizer's optimal value in \cref{eq:optimization} by
    \begin{equation*}
        \lambda T \rO
        +
        \sum_{t=1}^T \x{t}{} \qty(1 - \lambda \b{t}{})
        \le
        \rO T
        +
        \eta T
        +
        \frac{(1 + \rO)^2}{2 \eta}
        \le
        \rO T
        +
        \eta T
        +
        \frac{2}{\eta}
    \end{equation*}

    Setting $\eta = 1/\sqrt T$ bounds the number of wins of the Optimizer by $\rO T + 3 \sqrt{T}$.
\end{myproof}
\section{Interval Discrepancy with Multiple Learners} \label{sec:diff_b}

In the previous section, we proved that when an agent with a $\rho$ fraction of the budget follows the \cref{algo}, she is guaranteed to win at least $\rho T - \order*{\sqrt T}$ times.
In other words, in a system with $n$ agents who all follow our algorithm, no agent can gain more than $\order*{\sqrt T}$ wins by deviating to any other algorithm.
In this section, we perform a more fine-grained analysis of how these wins are distributed across time, specifically for first-price auctions.

We prove that when all agents run the \cref{algo}, after an initial convergence period, every agent with a $\rho$ fraction of the budget is guaranteed to win $\rho \tau - \order*{1}$ times in every interval of length $\tau \in \order*{\sqrt T}$.
The convergence period is $\order*{\sqrt T \log T}$, which means that it is very short.
In fact, we will show that agents' bids converge to $1$, which means that an agent $i$ with a starting bid $b_0$ needs $\Theta\qty\big(\frac{1 - b_0}{\eta \rho_i}) = \Theta\qty\big(\frac{(1 - b_0)\sqrt T}{\eta \rho_i})$ rounds to reach $1$, even if she loses all the time.
This is why our convergence time is proportional to $1/\rhomin$, where $\rhomin = \min_i \rho_i$.

\begin{restatable}{theorem}{DiscrepancyDiffBudgets}\label{thm:diff_b:main}
    Assume there are $n$ agents, with agent $i$ having a budget $\rho_i T$ with $\sum_i \rho_i = 1$.
    Assume that all agents are running the \cref{algo} with initial bids $\b{1}{i} \le \rho_i$ and step size $\eta$ and that $\eta \ll \rhomin$.
    Fix an interval $[t_1, t_2]$ of length $\tau$ such that $t_1 \ge \frac{6}{\rhomin\eta} \log \frac{1}{\eta}$.
    Then the number of wins that agent $i$ gets in that interval is at least
    \begin{equation*}
        \rho_i \tau 
        -
        \frac{6}{\rhomin} \rho_i \tau \eta
        -
        \frac{18}{\rhomin}
        .
    \end{equation*}
\end{restatable}

To prove this theorem, we will show that after an initial period, bids converge to being all within $\order*{\eta}$ close to $1$.
Once all agents' bids are this close to $1$, the discrepancy will follow (see \cref{lem:diff_b:wins_deviation}).
The technically challenging part of the proof is the convergence to this tight range of bids, which we will establish in \cref{ssec:diff_b:sqrt_eta,ssec:diff_b:eta_of_avg,ssec:diff_b:eta_of_one}.

Our tool that connects the discrepancy with the learning dynamics is the following lemma, which bounds the number of wins an agent gets as a function of how large or small her bids are.
The interesting regime of this theorem is when her bids are in an interval are in the range $[m, M]$ where both $m$ and $M$ are very close to $1$.
In particular, we will show that eventually $\abs{1 - M}, \abs{1 - m} \le \order{\eta}$, making the discrepancy of agent $i$ at most $\order{1}$ for short intervals, proving \cref{thm:diff_b:main}.

\begin{lemma} \label{lem:diff_b:wins_deviation}
    Fix an interval $[t_1, t_2)$ of length $\tau = t_2 - t_1$ and an agent $i$ who follows the \cref{algo}.
    Let $M = \max_{t \in [t_1, t_2]} \b{t}{i}$ and $m = \min_{t \in [t_1, t_2]} \b{t}{i}$ be the maximum and minimum bids of that agent in the interval.
    Let $\x{t}{i}$ be the indicator that agent $i$ wins round $t$.
    Then
    \begin{equation*}
        \sum_{t=t_1}^{t_2-1} \x{t}{i}
        \ge
        \frac{1}{M} \rho_i \tau
        -
        \frac{M - m}{\eta M}
        .
    \end{equation*}
\end{lemma}

\begin{myproof}
    Recall that for first-price auctions $\b{t+1}{i} = \b{t}{i} + \eta(\rho_i - \x{t}{i}\b{t}{i})$.
    Summing over all $t \in [t_1, t_2]$ and rearranging we get
    \begin{equation*}
        \sum_{t=t_1}^{t_2-1} \x{t}{i}\b{t}{i}
        =
        \rho_i \tau
        +
        \frac{\b{t_1}{i} - \b{t_2}{i}}{\eta}
        \ge
        \rho \tau
        -
        \frac{M - m}{\eta}
    \end{equation*}

    The lemma follows by $\x{t}{i}\b{t}{i} \le \x{t}{i} M$.
\end{myproof}

Using the above lemma, we prove \cref{thm:diff_b:main} by proving that $\b{t}{i} = 1 \pm \order*{\eta}$ for large enough $t$.
Specifically, the steps we follow are
\begin{enumerate}
    \item First, in \cref{ssec:diff_b:sqrt_eta}, we prove that after $\order*{\frac{1}{\eta} \log \frac{1}{\eta}}$ rounds, it holds that $\b{t}{i} = 1 \pm \order*{\sqrt\eta}$ for all $i$.
    Unfortunately, this falls short of \cref{thm:diff_b:main}: every agent is guaranteed only $\order*{1/\sqrt\eta} = \order*{T^{1/4}}$ discrepancy for small enough interval lengths $\tau$.
    
    \item Then, in \cref{ssec:diff_b:eta_of_avg}, we prove that after an additional $\order*{\frac{1}{\eta}}$ rounds, all the agents' bids are within $\order{\eta}$ of each other.
    While this does not prove anything useful for the discrepancy, it is a crucial step for the next result.

    \item Finally, in \cref{ssec:diff_b:eta_of_one}, we prove that after an additional $\order*{\frac{1}{\eta} \log \frac{1}{\eta}}$ rounds, all the bids are in the range $[1 - \order{\eta}, 1 + \order{\eta}]$, as needed.
    This will prove \cref{thm:adv:main}.
\end{enumerate}

We defer all the full proofs of the results of this section to \cref{sec:app:diff_b}.

\subsection{Convergence of bids to \texorpdfstring{$1 \pm \order*{\sqrt\eta}$}{1 +- sqrt(eta)}}
\label{ssec:diff_b:sqrt_eta}

In this section, we show that after an initial period of rounds, all the agents' bids converge to $1 \pm \order*{\sqrt\eta}$.
Our main technical tool is that the learning dynamics that the agents follow are the same as running subgradient descent with constant step $\eta$ on the function
\begin{equation}\label{eq:func_gradient}
    f(\vec b)
    =
    \frac{1}{2} \max_i b_i^2 - \vec\rho^\top \vec b + \frac{1}{2}
\end{equation}
where $\vec\rho = (\rho_i)_{i \in [n]}$.
We now prove some basic properties about the $f$ function, which will be useful to prove some kind of convergence.
We start with the fact that $f$ is convex and that one of its subgradients corresponds to the update step of our decentralized learning dynamics.

\begin{restatable}[Convexity of $f$]{proposition}{Convexity}\label{cl:convex}
    Function $f$ is convex and its subgradients for $\vec b \ge 0$ are
    \begin{equation*}
        \partial f(\vec b)
        =
        \textrm{conv}\qty(\qty{
            b_i \vec e_i - \vec\rho, i \in \argmax_k b_k
        })
    \end{equation*}
    where \textrm{conv} denotes the convex hull and $e_i$ is the unit vector in the $i$-th dimension.
    Note that one of the extreme points of $\partial f(\vec b)$ corresponds to the update vector of our learning dynamics.
\end{restatable}

We next show a crucial property that shows why we expect our dynamics to converge around the $\vec 1$ vector.
The function $f$ has a unique minimizer at $\vec b = \vec 1$.

\begin{restatable}[Minimum of $f$]{proposition}{Minimum}\label{cl:minimum}
    The unique minimum of $f$ is at $\vec b^\star = \vec 1$, with value $f(\vec b^\star) = 0$.
\end{restatable}

We now analyze the convergence of our dynamics.
Specifically, we examine how the distance from the optimum bid, $\norm*{\vecb{t}{} - \vec 1}$, evolves.
The following proposition uses the standard analysis of subgradient descent to show an upper bound of this distance in round $t+1$, given the distance of the previous round.
In particular, we will use this proposition to show that if $\norm*{\vecb{t}{} - \vec 1} \ge \Theta(\sqrt\eta)$, then the distance decreases.

\begin{restatable}{proposition}{NormTwoDecrease}\label{cl:diff_b:dist_from_one_decrease}
    Fix a round $t$ and let $\vec g^{(t)}$ be any subgradient of $f(\vecb{t}{})$. Then
    \begin{equation*}
        \norm{\vecb{t+1}{} - \vec 1}_2^2
        \le
        \qty( 1 - \eta \rhomin + 2 \eta^2)\norm{\vecb{t}{} - \vec 1}_2^2
        +
        3\eta^2
    \end{equation*}
    where $\rhomin = \min_i \rho_i$.
\end{restatable}

Much of the proof of the proposition follows standard analysis of subgradient descent.
A less standard property we use is \textit{quadratic growth} (see e.g. \cite{DBLP:journals/focm/DavisJ25}): the suboptimality gap of our function scales at least proportionally with the squared distance from the optimum.

While the above proposition does prove that $\norm*{\vecb{t}{} - \vec 1}_2^2$ decreases when the norm is big enough, it does not show a good enough convergence.
Specifically, if the inequality was an equality, the fixed point would be $\norm*{\vec b - \vec 1}_2^2 = \frac{3 \eta}{\rhomin - 2 \eta}$, which proves $\vec b = \vec 1 \pm \order*{\sqrt\eta}$ while our final goal is $\order*{\eta}$ deviation from $1$.
Nonetheless, we will use this bound to get our desired bound in the subsequent sections.
We now formally show that the norm converges to near the above fixed point.

\begin{restatable}{lemma}{ConvergenceSqrtEta}\label{lem:diff_b:sqrt_eta}
    Fix a round $t_0$ and assume that $\eta \le \rhomin/10$.
    Then, for rounds $t \ge t_0 + \frac{15}{4\eta\rhomin} \log\frac{1}{\eta}$ it holds that
    \begin{equation*}
        \norm{\vecb{t}{} - \vec 1}_2^2
        \le
        \frac{15 \eta}{4\rhomin}
        +
        \eta^3
        \norm{\vecb{t_0}{} - \vec 1}_2^2
    \end{equation*}
\end{restatable}

The proof of the lemma is fairly simple:  recursively applying \cref{cl:diff_b:dist_from_one_decrease}, we get an exponential decrease of the function up to the fixed point term.

\subsection{Convergence of bids to \texorpdfstring{$\order*{\eta}$}{O(eta)} of each other}
\label{ssec:diff_b:eta_of_avg}

\cref{lem:diff_b:sqrt_eta} proves that $\norm*{\vecb{t}{} - \vec 1} = \order*{\sqrt\eta}$ for large enough $t$.
This is an artifact of the quadratic growth property of $f$.
Unfortunately, $f$ fails to satisfy stronger properties like sharpness (see e.g. \cite{DBLP:journals/mor/CharisopoulosD24}) which require that the suboptimality gap grows at least proportionally with the distance from the optimum (and not the squared distance like in quadratic growth).
Such a property would indeed yield $\norm*{\vecb{t}{} - \vec 1} = \order*{\eta}$, as needed.

Since classic optimization techniques fail to get the required error, we examine how far the bids are from one another.
We show that all the bids get within $\order{\eta}$ of each other, which is the main goal of this subsection.
Specifically, we will show that $\norm*{\vec b - \bavg \vec 1}$ becomes $\order{\eta}$, where $\bavg = \frac{1}{n}\sum_i \b{}{i}$.
Our first claim concerns how this quantity decreases over time.
We prove a bound parametrized by the parameter $D$, where $D \sqrt \eta$ is an upper bound on the current distance of the bids from $1$.
When we combine all of our lemmas later, we will use this lemma for a value of $D$ to be decided later.

\begin{restatable}{proposition}{DistFromAvgDecrease} \label{cl:diff_b:dist_from_avg_decrease}
    For a parameter $D>0$, suppose $t$ is a round such that $\norm*{\vecb{t}{} - \vec 1}_2 \le D \sqrt\eta
    \le \sqrt{\frac{1}{2n}}$.
    Then it holds
    \begin{equation*}
        \norm{\vecb{t+1}{} - \bavg^{(t+1)}\vec 1}_2^2
        \le
        \norm{\vecb{t}{} - \bavg^{(t)}\vec 1}_2^2
        -
        2 \rhomin \eta \norm{\vecb{t}{} - \bavg^{(t)}\vec 1}_2
        +
        4 \eta^2 \qty( D^2 + 1 ) .
    \end{equation*}
\end{restatable}


We point out two things about this proposition.
First, the decrease depends on the distance $\norm*{\vecb{t}{} - \vec 1}_2$: if this was $\omega(\sqrt\eta)$, then the constant term would be $\omega(\eta^2)$ and we wouldn't get the desired error on $\norm*{\vecb{}{} - \bavg \vec 1}_2$.
This shows why we first showed this distance is small in \cref{ssec:diff_b:sqrt_eta}.
Second, the middle term, that is proportional to the unsquared distance $\norm*{\vec b - \bavg\vec 1}_2$, is reminiscent of what we would get for $\norm*{\vec b - \vec 1}$ if the sharpness property were true (that the suboptimality gap is at least proportional to the distance from the optimal).

We notice that this claim should prove that the bids are within $\order*{\eta}$ of one another: the fixed point is $\norm{\vec b - \bavg\vec 1}_2 = 2 \eta \frac{D^2 + 1}{\rhomin}$.
With the next lemma, we formalize this convergence.

\begin{restatable}{lemma}{ConvergenceDistFromAvg}\label{lem:diff_b:dist_from_avg_eta}
    Fix a round $t_0$ and a parameter $D>0$, such that for $t \ge t_0$, $\norm*{\vecb{t}{} - \vec 1}_2 \le D \sqrt\eta \le \sqrt{\frac{1}{2n}}$.
    Then, for $t \ge t_0 + 1 + \frac{1}{2\eta}$ it holds that
    \begin{equation*}
        \norm{\vecb{t}{} - \bavg^{(t)}\vec 1}_2
        \le
        3 \frac{D^2 + 1}{\rhomin} \eta
    \end{equation*}
\end{restatable}

The convergence of this is much faster than the one in \cref{lem:diff_b:sqrt_eta}.
This is because, when $\norm*{\vec b - \vec 1}$ is larger than the target, i.e., $\norm*{\vec b - \vec 1} \ge 3 \frac{D^2 + 1}{\rhomin} \eta$, then the distances decreases by at least a constant.
This means that it only needs $\order*{1/\eta}$ steps to reach the target.

\subsection{Convergence of bids \texorpdfstring{$\order*{\eta}$}{O(eta)} close to \texorpdfstring{$1$}{1}}
\label{ssec:diff_b:eta_of_one}

Now that we know that the bids are $\order*{\eta}$ close to the average bid, we prove that the average bid converges close to $1$.
This is the easiest part of the argument, because the following update rule holds: $\bavg^{(t+1)} = \bavg^{(t)} + \frac{\eta}{n}\qty\big(1 - \bmax^{(t)})$.
By $\norm*{\vec b - \bavg\vec 1}_2 = \order*{\eta}$, we get that $\bavg = \bmax \pm \order*{\eta}$, which makes it easy to show that the average bid converges close to $1$.

\begin{restatable}{lemma}{ConvergenceDistFromOneEta}\label{lem:diff_b:eta_dist_from_one}
    Fix a round $t_0$ and a parameter $D>0$, such that for $t \ge t_0$, $\norm*{\vecb{t}{} - \vec 1}_2 \le D \sqrt\eta \le \sqrt{\frac{1}{2n}}$ and $\norm*{\vec b^{(t)} - \bavg^{(t)}\vec 1}_2 \le A \eta$, where $A = 3 \frac{D^2 + 1}{\rhomin}$.
    Also assume that $\eta \le 1/2$.
    Then for any $t \ge t_0 + 3\frac{n}{\eta}\log\frac{1}{\eta}$ and for any $i\in[n]$ it holds
    \begin{equation*}
        - 2A \eta - \eta^2
        \le \b{t}{i} - 1 \le
        A\eta
        +
        2 \eta^3
    \end{equation*}
\end{restatable}

\subsection{Putting everything together}

We now combine all the lemmas from \cref{ssec:diff_b:sqrt_eta,ssec:diff_b:eta_of_avg,ssec:diff_b:eta_of_one} to get the exact statement for how long it will take for all the bids to be $1 \pm \order*{\eta}$.

\begin{restatable}{theorem}{FinalDistFromOne}\label{thm:diff_b:bids_are_one_pm_eta}
    Fix a round $t \ge \frac{11}{\eta \rhomin} \log\frac{1}{\eta}$.
    Assume that $\max_i \abs{\b{1}{i} - 1} \le \frac{1}{\sqrt{n \eta}}$ and that $\eta \le \frac{\rhomin^5}{667}$.
    Then, for all $i$ it holds that
    \begin{equation*}
         1 - \frac{12}{\rhomin} \eta - \eta^2
         \le \b{t}{i} \le
         1 + \frac{6}{\rhomin} \eta + 2\eta^3
    \end{equation*}
\end{restatable}


We notice that our initial result of \cref{lem:diff_b:sqrt_eta} suggests using \cref{lem:diff_b:dist_from_avg_eta,lem:diff_b:eta_dist_from_one} by setting the parameter $D$ to be $\frac{15}{4\rhomin}$.
This would make the parameter $A$ in \cref{lem:diff_b:eta_dist_from_one} to be $\frac{45}{4\rhomin^2}$, implying a proportional bound for the above theorem.
However, our result above obtains better dependence, both in terms of the absolute constant and the dependence on $\rhomin$.
This is because, instead of three convergence periods---one for each \cref{lem:diff_b:sqrt_eta,lem:diff_b:dist_from_avg_eta,lem:diff_b:eta_dist_from_one}---we apply five convergence periods, using the final two lemmas another time.
By using the first three lemmas with parameter $D$ set to $\frac{15}{4\rhomin}$, we get that eventually $\norm*{\vecb{}{} - \vec 1 }_2 \le \order{\eta}$.
This allows us to re-apply \cref{lem:diff_b:dist_from_avg_eta,lem:diff_b:eta_dist_from_one}, this time setting parameter $D$ to $\order*{\sqrt\eta}$ and $A$ to $\frac{3}{\rhomin} + \order{\eta}$, which yields the better bound we get in \cref{thm:diff_b:bids_are_one_pm_eta}.

\cref{thm:diff_b:bids_are_one_pm_eta}, combined with \cref{lem:diff_b:wins_deviation} easily yields \cref{thm:diff_b:main}.

\section{Interval Discrepancy with Equal Budgets} \label{sec:same_b}

In this section, we study the special case of \cref{sec:diff_b}, where the $n$ players running the \cref{algo} all have the same budget: $\rho_i = \rho = 1/n$.
As we will see from our following results, this setting is much easier.
One reason is that when agents have different budgets, even a central coordinator allocating wins cannot guarantee every agent $i$ winning once every $1/\rho_i$ rounds.
This can become apparent when the budget shares are $\frac{1}{2}, \frac{1}{3}, \frac{1}{6}$: no schedule of wins can guarantee these agents to win once every $2$, $3$, and $6$ rounds, respectively.
However, this is not the case with $n$ agents with equal budgets: the round robin schedule where every agent wins once every $n$ is achievable.
This is exactly what we prove happens when these $n$ agents run the \cref{algo}.
Similar to \cref{sec:diff_b}, we do not necessarily assume that agents start at the same bid; in fact, for equal budgets, if the bids are all equal, \cref{lem:same_b:round_robin_condition} proves that agents start doing round-robin after $n$ rounds.

\begin{theorem}\label{thm:same_b:round_robin}
    Fix $n$ agents with equal budgets and assume all of them are running the \cref{algo} with $\eta$ and starting bid satisfying $\bmax^{(1)} \le \rho = \frac{1}{n}$.
    Then, after round $\frac{n^2}{2\eta} \log \frac{1}{\eta} + n + 1$, all the agents start winning in a round-robin fashion.
    In other words, in any time interval $[t_1, t_2]$ with $t_1 \ge \frac{n^2}{2\eta} \log \frac{1}{\eta} + n + 1$, the discrepancy of every agent is at most $\frac{n-1}{n}$.
\end{theorem}

To prove the theorem, we give a characterization of what condition the bids need to satisfy to get the desired round-robin behavior.
It turns out that this condition is very simple: if the bids are strictly ordered and also the maximum bid is not much higher than the minimum bid, then the agents round-robin forever.
The strict ordering of the bids is required so that our condition is robust against any tie-breaking.
The condition between the maximum and the minimum bid is also simple: it ensures that the agent with the highest bid will be the minimum bid in the next round.

\begin{lemma}\label{lem:same_b:round_robin_condition}
    Fix a round $t$ and assume $0 < \eta < 1$.
    If there exists a permutation of the agents $\pi: [n] \to [n]$ such that
    \begin{equation}\label{eq:641}
        \b{t}{\pi(1)} > \b{t}{\pi(2)} > \ldots > \b{t}{\pi(n-1)} > \b{t}{\pi(n)} > (1 - \eta) \b{t}{\pi(1)}
    \end{equation}
    then from round $t$ onward, agents perform round-robin.
\end{lemma}

Instead of directly proving the above lemma, we will prove a slightly stronger claim.
This claim shows that \eqref{eq:641} holding in round $t$ implies that a similar inequality holds in round $t+1$, even if the inequalities can be equalities.
In addition, it shows that the number of strict inequalities is (in a way) strictly increasing, which means that eventually all the inequalities will be strict.
This will also help us prove convergence to the desired condition.

\begin{proposition}\label{claim:same_b:ordering_of_bids_condition}
    Fix a round $t$ and assume $0 < \eta < 1$.
    Assume that $\max_i \b{t}{i} \ge (1 - \eta) \min_i \b{t}{i}$.
    Fix a permutation of the agents $\pi: [n] \to [n]$ such that
    \begin{equation*}
        \b{t}{\pi(1)} \ge \b{t}{\pi(2)} \ge \ldots \ge \b{t}{\pi(n-1)} \ge \b{t}{\pi(n)} \ge (1 - \eta) \b{t}{\pi(1)}
    \end{equation*}
    where $\pi(1)$ is the winner of round $t$ (under arbitrary tie-breaking).
    Assume that the last $k$ inequalities are strict for some $k=0,1,2,\ldots,n$.
    Then, it holds 
    \begin{equation*}
        \b{t+1}{\pi(2)} \ge \b{t+1}{\pi(3)} \ge \ldots \ge \b{t+1}{\pi(n)} \ge \b{t+1}{\pi(1)} \ge (1 - \eta) \b{t+1}{\pi(2)}
    \end{equation*}
    where the last $\min(n, k+1)$ inequalities are strict.
\end{proposition}

\begin{myproof}
    W.l.o.g. re-order the agents so that $\pi(i) = i$.
    Fix a $t < T$, where agent $1$ wins round $t$.
    We now calculate the bids of the next round: $\b{t+1}{1} = (1-\eta) \b{t}{1} + \eta \rho$ and $\b{t+1}{i} = \b{t}{i} + \eta \rho$, for $i \ge 2$.
    We have that
    \begin{equation*}
    \begin{gathered}
        \b{t+1}{2} \ge \b{t+1}{3} \ge \ldots \ge \b{t+1}{n} \ge \b{t+1}{1} \ge (1 - \eta) \b{t+1}{2}
        \\\iff\\
        \qquad\quad\b{t}{2} \ge \b{t}{3} \ge \ldots \ge \b{t}{n} \ge (1-\eta) \b{t}{1} \ge (1 - \eta) \b{t}{2} - \eta^2 \rho
    \end{gathered}
    \end{equation*}
    We notice that by our assumption on the bids of round $t$, the first $n-1$ inequalities are true and the last $k$ of them are strict inequalities.
    All we have to show know is that the last inequality is true and is strict.
    This is equivalent with $\b{t}{1} > \b{t}{2} - \frac{\eta^2}{1-\eta} \rho$ which is implied by $\b{t}{1} \ge \b{t}{2}$ and $0 < \eta < 1$.
\end{myproof}

We now easily prove \cref{lem:same_b:round_robin_condition}.

\begin{myproof}[Proof of \cref{lem:same_b:round_robin_condition}]
    Follows by \cref{claim:same_b:ordering_of_bids_condition}:
    for every round $t' \ge t$ we get that (define $\pi(n+1) = \pi(1)$, $\pi(n+2) = \pi(2)$, etc.) such that,
    $$\b{t'}{\pi(1 + t' - t)} > \b{t'}{\pi(2 + t' - t)} > \ldots > \b{t'}{\pi(n + t' - t)}$$ which guarantees round-robin.
\end{myproof}

To complete the proof of \cref{thm:same_b:round_robin} all we need to examine is in which round $\bmin^{(t)} \ge (1-\eta) \bmax^{(t)}$; by \cref{claim:same_b:ordering_of_bids_condition}, after an additional $n$ rounds we get the condition of \cref{lem:same_b:round_robin_condition}.

\begin{lemma}\label{lem:same_b:convergence}
    Assume $0 < \eta < 1$ and that $\sum_i\b{1}{i} \le n$.
    Then, for every round $t \ge 1 + \frac{n^2}{2\eta} \log\frac{1}{\eta}$ it holds that $\bmin^{(t)} \ge (1-\eta) \bmax^{(t)}$.
\end{lemma}

\begin{myproof}
    We first establish a global upper bound on $\bmax^{(t)}$, based on $n$.
    We will prove that $\sum_i \b{t}{i} \le n$, which implies the same upper bound for $\bmax^{(t)}$.
    The above holds for $t = 1$.
    For an arbitrary $t \ge 1$, by summing the update rules of all players,
    \begin{equation*}
        \sum_i \b{t+1}{i}
        =
        \sum_i \b{t}{i}
        +
        \eta \qty(1 - \bmax^{(t)})
        \le
        \sum_i \b{t}{i}
        +
        \eta \qty(1 - \frac{1}{n}\sum_i \b{t}{i})
        \le
        \qty(1 - \frac{\eta}{n}) n
        +
        \eta
        = n
    \end{equation*}
    where the last inequality follows from induction that $\sum_i \b{t}{i} \le n$.
    This proves that $\bmax^{(t)} \le n$.

    Now we study how the maximum and minimum bids evolve over time:
    \begin{align*}
        \bmax^{(t+1)} &\le \max\qty{ (1 - \eta)\bmax^{(t)} , \bmax^{(t)} } + \eta \rho = \bmax^{(t)}  + \eta \rho
        \\\textrm{ and }\quad
        \bmin^{(t+1)} &= \min\qty{ (1 - \eta)\bmax^{(t)} , \bmin^{(t)} } + \eta \rho
    \end{align*}

    Now let $t_0$ be a round where $\bmin^{(t)} < (1-\eta) \bmax^{(t)}$ for all $t \le t_0$.
    We will show that such a $t_0$ has to be small.
    Fix such any $t \le t_0$.
    Then,
    \begin{equation*}
        \frac{\bmin^{(t+1)}}{\bmax^{(t+1)}}
        \ge
        \frac{\bmin^{(t)} + \eta \rho}{\bmax^{(t)} + \eta \rho}
        =
        \frac{
            \frac{\bmin^{(t)}}{\bmax^{(t)}} + \frac{\eta \rho}{\bmax^{(t)}}
        }{
            1 + \frac{\eta \rho}{\bmax^{(t)}}
        }
        \ge
        \frac{
            \frac{\bmin^{(t)}}{\bmax^{(t)}} + \eta \rho^2
        }{
            1 + \eta \rho^2
        }
    \end{equation*}
    where in the last inequality follows by the fact that $x \mapsto \frac{y + x}{1 + x}$ is increasing for $y < 1$ and that $\bmax^{(t)} \le 1/\rho$.
    Let $r^{(t)} = \frac{\bmin^{(t)}}{\bmax^{(t)}}$; using induction, the above shows that for every $t \le t_0 + 1$
    \begin{equation*}
        r^{(t)}
        \ge
        1 - \qty( \frac{1}{1 + \eta \rho^2} )^{t-1}
        +
        \qty( \frac{1}{1 + \eta \rho^2} )^{t-1} r^{(1)}
        \ge
        1 - \qty( \frac{1}{1 + \eta \rho^2} )^{t-1}
    \end{equation*}

    We now use the above inequality for $t = t_0$, along with the fact that $r^{(t_0)} < 1 - \eta$ to get
    \begin{equation*}
        1 - \eta
        >
        1 - \qty( \frac{1}{1 + \eta \rho^2} )^{t_0-1}
        \quad\implies\quad
        t_0
        <
        1 + \frac{ \log\frac{1}{\eta} }{\log\qty(1 + \eta \rho^2)}
        \le
        1 + \frac{ \log\frac{1}{\eta} }{2\eta \rho^2}
    \end{equation*}
    where in the last inequality we use that $\log(1 + x) \ge x/2$ for $0 \le x \le 1$.
    This proves that for rounds $t \ge t_0 + 1$ it holds that $\bmin^{(t)} \ge (1-\eta) \bmax^{(t)}$ (by \cref{claim:same_b:ordering_of_bids_condition} if this condition holds for some round it holds for all subsequent rounds).
\end{myproof}

We now finally get \cref{thm:same_b:round_robin} as a corollary from \cref{lem:same_b:round_robin_condition,lem:same_b:convergence}.


\printbibliography{}

@article{aggarwal2024auto,
  title     = {Auto-bidding and Auctions in Online Advertising: A Survey},
  author    = {Aggarwal, Gagan and Badanidiyuru, Ashwinkumar and Balseiro, Santiago R and Bhawalkar, Kshipra and Deng, Yuan and Feng, Zhe and Goel, Gagan and Liaw, Christopher and Lu, Haihao and Mahdian, Mohammad and others},
  journal   = {ACM SIGecom Exchanges},
  volume    = {22},
  number    = {1},
  pages     = {159--183},
  year      = {2024},
  publisher = {ACM New York, NY, USA},
  address   = {New York, NY, USA}
}

@inproceedings{braverman2018selling,
  title     = {Selling to a no-regret buyer},
  author    = {Braverman, Mark and Mao, Jieming and Schneider, Jon and Weinberg, Matt},
  booktitle = {Proceedings of the 2018 ACM Conference on Economics and Computation},
  pages     = {523--538},
  publisher = {Association for Computing Machinery},
  address   = {New York, NY, USA},
  year      = {2018}
}

@article{broadbent1993advertising,
  title     = {Advertising effects: More than short term},
  author    = {Broadbent, Simon},
  journal   = {Market Research Society. Journal.},
  volume    = {35},
  number    = {1},
  pages     = {1--11},
  year      = {1993},
  publisher = {SAGE Publications Sage UK: London, England}
}

@article{broadbent1995adstock,
  title     = {Adstock modelling for the long term},
  author    = {Broadbent, Simon and Fry, Tim},
  journal   = {Market Research Society. Journal.},
  volume    = {37},
  number    = {4},
  pages     = {1--18},
  year      = {1995},
  publisher = {SAGE Publications Sage UK: London, England}
}

@article{broadbent2000advertisements,
  title     = {What do advertisements really do for brands?},
  author    = {Broadbent, Simon},
  journal   = {International Journal of Advertising},
  volume    = {19},
  number    = {2},
  pages     = {147--165},
  year      = {2000},
  publisher = {Taylor \& Francis}
}

@inproceedings{cai2023selling,
  title        = {Selling to multiple no-regret buyers},
  author       = {Cai, Linda and Weinberg, S Matthew and Wildenhain, Evan and Zhang, Shirley},
  booktitle    = {International Conference on Web and Internet Economics},
  pages        = {113--129},
  publisher    = {Springer-Verlag},
  address      = {Berlin, Heidelberg},
  year         = {2023},
  organization = {Springer}
}

@book{cesa2006prediction,
  title     = {Prediction, learning, and games},
  author    = {Cesa-Bianchi, Nicolo and Lugosi, G{\'a}bor},
  year      = {2006},
  publisher = {Cambridge university press},
  address   = {USA}
}

@article{craig1976advertising,
  title     = {Advertising wearout: An experimental analysis},
  author    = {Craig, C Samuel and Sternthal, Brian and Leavitt, Clark},
  journal   = {Journal of Marketing Research},
  volume    = {13},
  number    = {4},
  pages     = {365--372},
  year      = {1976},
  publisher = {SAGE Publications Sage CA: Los Angeles, CA}
}

@inproceedings{DBLP:conf/colt/LucierPSZ24,
  author    = {Brendan Lucier and
               Sarath Pattathil and
               Aleksandrs Slivkins and
               Mengxiao Zhang},
  title     = {Autobidders with Budget and {ROI} Constraints: Efficiency, Regret,
               and Pacing Dynamics},
  booktitle = {The Thirty Seventh Annual Conference on Learning Theory, June 30 - July 3, 2023, Edmonton, Canada},
  series    = {Proceedings of Machine Learning Research},
  volume    = {247},
  pages     = {3642--3643},
  publisher = {{PMLR}},
  year      = {2024},
  address   = {Edmonton, Canada}
}

@inproceedings{DBLP:conf/innovations/GaitondeLLLS23,
  author    = {Jason Gaitonde and
               Yingkai Li and
               Bar Light and
               Brendan Lucier and
               Aleksandrs Slivkins},
  title     = {Budget Pacing in Repeated Auctions: Regret and Efficiency Without
               Convergence},
  booktitle = {14th Innovations in Theoretical Computer Science Conference, {ITCS}
               2023, January 10-13, 2023, MIT, Cambridge, Massachusetts, {USA}},
  series    = {LIPIcs},
  volume    = {251},
  pages     = {52:1--52:1},
  publisher = {Schloss Dagstuhl - Leibniz-Zentrum f{\"{u}}r Informatik},
  year      = {2023},
  address   = {Cambridge, Massachusetts, {USA}}
}

@inproceedings{DBLP:conf/wine/ConitzerKSM18,
  author    = {Vincent Conitzer and
               Christian Kroer and
               Eric Sodomka and
               Nicol{\'{a}}s E. Stier Moses},
  editor    = {George Christodoulou and
               Tobias Harks},
  title     = {Multiplicative Pacing Equilibria in Auction Markets},
  booktitle = {Web and Internet Economics - 14th International Conference, {WINE}
               2018, Oxford, UK, December 15-17, 2018, Proceedings},
  series    = {Lecture Notes in Computer Science},
  volume    = {11316},
  pages     = {443},
  publisher = {Springer},
  year      = {2018},
  address   = {Linthicum, MD, USA}
}

@inproceedings{DBLP:journals/corr/AggarwalFZ24,
  author    = {Gagan Aggarwal and
               Giannis Fikioris and
               Mingfei Zhao},
  title     = {No-Regret Algorithms in non-Truthful Auctions with Budget and {ROI}
               Constraints},
  booktitle = {Proceedings of the ACM Web Conference 2025},
  publisher = {Association for Computing Machinery},
  address   = {New York, NY, USA},
  year      = {2025},
  url       = {https://openreview.net/forum?id=qsj78d8y1j},
  numpages  = {}
}

@article{DBLP:journals/mansci/BalseiroG19,
  author  = {Santiago R. Balseiro and
             Yonatan Gur},
  title   = {Learning in Repeated Auctions with Budgets: Regret Minimization and
             Equilibrium},
  journal = {Manag. Sci.},
  volume  = {65},
  number  = {9},
  pages   = {3952--3968},
  year    = {2019}
}

@article{fikioris2025mor,
  author       = {Giannis Fikioris and
                  {\'{E}}va Tardos},
  title        = {Liquid Welfare Guarantees for No-Regret Learning in Sequential Budgeted
                  Auctions},
  journal      = {Math. Oper. Res.},
  volume       = {50},
  number       = {2},
  pages        = {1233--1249},
  year         = {2025},
  url          = {https://doi.org/10.1287/moor.2023.0274},
  doi          = {10.1287/MOOR.2023.0274},
  bibsource    = {dblp computer science bibliography, https://dblp.org}
}

@article{ha1997does,
  title     = {Does advertising clutter have diminishing and negative returns?},
  author    = {Ha, Louisa and Litman, Barry R},
  journal   = {Journal of Advertising},
  volume    = {26},
  number    = {1},
  pages     = {31--42},
  year      = {1997},
  publisher = {Taylor \& Francis}
}

@inproceedings{rubinstein2024strategizing,
  author    = {Rubinstein, Aviad and Zhao, Junyao},
  title     = {Strategizing against No-Regret Learners in First-Price Auctions},
  year      = {2024},
  publisher = {Association for Computing Machinery},
  address   = {New York, NY, USA},
  booktitle = {Proceedings of the 25th ACM Conference on Economics and Computation},
  pages     = {894–921},
  numpages  = {28},
  location  = {New Haven, CT, USA},
  series    = {EC '24}
}

@article{weinberg1982econometric,
  title     = {On the econometric measurement of the duration of advertising effect on sales},
  author    = {Weinberg, Charles B and Weiss, Doyle L},
  journal   = {Journal of Marketing Research},
  volume    = {19},
  number    = {4},
  pages     = {585--591},
  year      = {1982},
  publisher = {SAGE Publications Sage CA: Los Angeles, CA}
}

@inproceedings{DBLP:conf/sigecom/FikiorisKKKMT25,
  author       = {Giannis Fikioris and
                  Robert Kleinberg and
                  Yoav Kolumbus and
                  Raunak Kumar and
                  Yishay Mansour and
                  {\'{E}}va Tardos},
  editor       = {Itai Ashlagi and
                  Aaron Roth},
  title        = {Learning in Budgeted Auctions with Spacing Objectives},
  booktitle    = {Proceedings of the 26th {ACM} Conference on Economics and Computation,
                  {EC} 2025, Stanford University, Stanford, CA, USA, July 7-10, 2025},
  pages        = {158},
  publisher    = {{ACM}},
  year         = {2025},
  url          = {https://doi.org/10.1145/3736252.3742512},
  doi          = {10.1145/3736252.3742512},
  bibsource    = {dblp computer science bibliography, https://dblp.org}
}

@inproceedings{DBLP:conf/colt/MansourMSS22,
  author       = {Yishay Mansour and
                  Mehryar Mohri and
                  Jon Schneider and
                  Balasubramanian Sivan},
  editor       = {Po{-}Ling Loh and
                  Maxim Raginsky},
  title        = {Strategizing against Learners in Bayesian Games},
  booktitle    = {Conference on Learning Theory, 2-5 July 2022, London, {UK}},
  series       = {Proceedings of Machine Learning Research},
  volume       = {178},
  pages        = {5221--5252},
  publisher    = {{PMLR}},
  year         = {2022},
  url          = {https://proceedings.mlr.press/v178/mansour22a.html},
  bibsource    = {dblp computer science bibliography, https://dblp.org}
}

@inproceedings{DBLP:conf/nips/DengSS19,
  author       = {Yuan Deng and
                  Jon Schneider and
                  Balasubramanian Sivan},
  editor       = {Hanna M. Wallach and
                  Hugo Larochelle and
                  Alina Beygelzimer and
                  Florence d'Alch{\'{e}}{-}Buc and
                  Emily B. Fox and
                  Roman Garnett},
  title        = {Strategizing against No-regret Learners},
  booktitle    = {Advances in Neural Information Processing Systems 32: Annual Conference
                  on Neural Information Processing Systems 2019, NeurIPS 2019, December
                  8-14, 2019, Vancouver, BC, Canada},
  pages        = {1577--1585},
  year         = {2019},
  url          = {https://proceedings.neurips.cc/paper/2019/hash/8b6dd7db9af49e67306feb59a8bdc52c-Abstract.html},
  bibsource    = {dblp computer science bibliography, https://dblp.org}
}

@article{hart2000simple,
  author    = {Hart, Sergiu and Mas{-}Colell, Andreu},
  title     = {A Simple Adaptive Procedure Leading to Correlated Equilibrium},
  journal   = {Econometrica},
  volume    = {68},
  number    = {5},
  pages     = {1127--1150},
  year      = {2000},
  publisher = {Wiley}
}

@article{DBLP:journals/corr/abs-1909-05207,
  author       = {Elad Hazan},
  title        = {Introduction to Online Convex Optimization},
  journal      = {CoRR},
  volume       = {abs/1909.05207},
  year         = {2019},
  url          = {http://arxiv.org/abs/1909.05207},
  eprinttype    = {arXiv},
  eprint       = {1909.05207},
  bibsource    = {dblp computer science bibliography, https://dblp.org}
}

@article{DBLP:journals/corr/abs-2202-06152,
  author       = {Santiago R. Balseiro and
                  Haihao Lu and
                  Vahab S. Mirrokni and
                  Balasubramanian Sivan},
  title        = {Analysis of Dual-Based PID Controllers through Convolutional Mirror Descent},
  journal      = {CoRR},
  volume       = {abs/2202.06152},
  year         = {2022},
  url          = {https://arxiv.org/abs/2202.06152},
  eprinttype    = {arXiv},
  eprint       = {2202.06152},
  bibsource    = {dblp computer science bibliography, https://dblp.org}
}

@inproceedings{DBLP:conf/icml/WangYDK23,
  author       = {Qian Wang and
                  Zongjun Yang and
                  Xiaotie Deng and
                  Yuqing Kong},
  editor       = {Andreas Krause and
                  Emma Brunskill and
                  Kyunghyun Cho and
                  Barbara Engelhardt and
                  Sivan Sabato and
                  Jonathan Scarlett},
  title        = {Learning to Bid in Repeated First-Price Auctions with Budgets},
  booktitle    = {International Conference on Machine Learning, {ICML} 2023, 23-29 July
                  2023, Honolulu, Hawaii, {USA}},
  series       = {Proceedings of Machine Learning Research},
  volume       = {202},
  pages        = {36494--36513},
  publisher    = {{PMLR}},
  year         = {2023},
  url          = {https://proceedings.mlr.press/v202/wang23ao.html},
  bibsource    = {dblp computer science bibliography, https://dblp.org}
}

@inproceedings{DBLP:conf/sigecom/KumarSS24,
  author       = {Rachitesh Kumar and
                  Jon Schneider and
                  Balasubramanian Sivan},
  editor       = {Dirk Bergemann and
                  Robert Kleinberg and
                  Daniela Sab{\'{a}}n},
  title        = {Strategically-Robust Learning Algorithms for Bidding in First-Price
                  Auctions},
  booktitle    = {Proceedings of the 25th {ACM} Conference on Economics and Computation,
                  {EC} 2024, New Haven, CT, USA, July 8-11, 2024},
  pages        = {893},
  publisher    = {{ACM}},
  year         = {2024},
  url          = {https://doi.org/10.1145/3670865.3673514},
  doi          = {10.1145/3670865.3673514},
  bibsource    = {dblp computer science bibliography, https://dblp.org}
}

@inproceedings{DBLP:conf/sigecom/Arunachaleswaran25,
  author       = {Eshwar Ram Arunachaleswaran and
                  Natalie Collina and
                  Yishay Mansour and
                  Mehryar Mohri and
                  Jon Schneider and
                  Balasubramanian Sivan},
  editor       = {Itai Ashlagi and
                  Aaron Roth},
  title        = {Swap Regret and Correlated Equilibria Beyond Normal-Form Games},
  booktitle    = {Proceedings of the 26th {ACM} Conference on Economics and Computation,
                  {EC} 2025, Stanford University, Stanford, CA, USA, July 7-10, 2025},
  pages        = {130--157},
  publisher    = {{ACM}},
  year         = {2025},
  url          = {https://doi.org/10.1145/3736252.3742511},
  doi          = {10.1145/3736252.3742511},
  bibsource    = {dblp computer science bibliography, https://dblp.org}
}

@inproceedings{DBLP:conf/sigecom/Arunachaleswaran24,
  author       = {Eshwar Ram Arunachaleswaran and
                  Natalie Collina and
                  Jon Schneider},
  editor       = {Dirk Bergemann and
                  Robert Kleinberg and
                  Daniela Sab{\'{a}}n},
  title        = {Pareto-Optimal Algorithms for Learning in Games},
  booktitle    = {Proceedings of the 25th {ACM} Conference on Economics and Computation,
                  {EC} 2024, New Haven, CT, USA, July 8-11, 2024},
  pages        = {490--510},
  publisher    = {{ACM}},
  year         = {2024},
  url          = {https://doi.org/10.1145/3670865.3673517},
  doi          = {10.1145/3670865.3673517},
  bibsource    = {dblp computer science bibliography, https://dblp.org}
}

@article{Hannan57,
    author = {James F. Hannan},
    title = {Approximation to Bayes risk in repeated play},
    journal = {Contributions
to the Theory of Games},
    year = 1957,
    pages ={97--139}
}

@article{DBLP:journals/focm/DavisJ25,
  author       = {Damek Davis and
                  Liwei Jiang},
  title        = {A Local Nearly Linearly Convergent First-Order Method for Nonsmooth
                  Functions with Quadratic Growth},
  journal      = {Found. Comput. Math.},
  volume       = {25},
  number       = {3},
  pages        = {943--1024},
  year         = {2025},
  url          = {https://doi.org/10.1007/s10208-024-09653-y},
  doi          = {10.1007/S10208-024-09653-Y},
  bibsource    = {dblp computer science bibliography, https://dblp.org}
}

@article{DBLP:journals/mor/CharisopoulosD24,
  author       = {Vasileios Charisopoulos and
                  Damek Davis},
  title        = {A Superlinearly Convergent Subgradient Method for Sharp Semismooth
                  Problems},
  journal      = {Math. Oper. Res.},
  volume       = {49},
  number       = {3},
  pages        = {1678--1709},
  year         = {2024},
  url          = {https://doi.org/10.1287/moor.2023.1390},
  doi          = {10.1287/MOOR.2023.1390},
  bibsource    = {dblp computer science bibliography, https://dblp.org}
}

@article{DBLP:journals/jacm/BadanidiyuruKS18,
  author       = {Ashwinkumar Badanidiyuru and
                  Robert Kleinberg and
                  Aleksandrs Slivkins},
  title        = {Bandits with Knapsacks},
  journal      = {J. {ACM}},
  volume       = {65},
  number       = {3},
  pages        = {13:1--13:55},
  year         = {2018},
  url          = {https://doi.org/10.1145/3164539},
  doi          = {10.1145/3164539},
  bibsource    = {dblp computer science bibliography, https://dblp.org}
}

@inproceedings{Guruganesh2024,
 author = {Guru Guruganesh and Kolumbus, Yoav and Schneider, Jon and Talgam-Cohen, Inbal and Vlatakis-Gkaragkounis, Emmanouil-Vasileios and Wang, Joshua R. and Weinberg, S. Matthew},
 booktitle = {Advances in Neural Information Processing Systems},
 editor = {A. Globerson and L. Mackey and D. Belgrave and A. Fan and U. Paquet and J. Tomczak and C. Zhang},
 pages = {77366--77408},
 publisher = {Curran Associates, Inc.},
 title = {Contracting with a Learning Agent},
 url = {https://proceedings.neurips.cc/paper_files/paper/2024/file/8d7c8a3a0ed04006d129b3cebcac7a3e-Paper-Conference.pdf},
 volume = {37},
 year = {2024}
}

@article{lin2024generalized,
  title={Generalized principal-agent problem with a learning agent},
  author={Lin, Tao and Chen, Yiling},
  journal={arXiv preprint arXiv:2402.09721},
  year={2024}
}

\newpage
\appendix
\section{Missing proofs of Section~\ref{sec:adv}} \label{sec:app:model}

We begin by stating the full version of \cref{lem:adv:opt_lagrange_normalized}.

\begin{lemma}[Stronger version of \cref{lem:adv:opt_lagrange_normalized}] \label{lem:adv:opt_lagrange_unnormalized}
    Fix $\lambda = \frac{\rL}{(\rL + \rO)^2}$ and an interval $[T-\tau + 1, T]$ of length $\tau$ for any $\tau \in [0, T]$ and the Learner's initial bid in that interval $b = \b{T-\tau+1}{}$.
    Then, for any binary sequence $\{\x{t}{}\}_{t \in [T-\tau+1, T]}$ it holds
    \begin{equation*}
        \sum_{t=T-\tau + 1}^T \x{t}{} (1 - \lambda \b{t}{})
        \le
        \frac{\rO^2 + \qty(\rO^2 + \rL^2)\eta}{(\rO + \rL)^2}
        \tau
        +
        g(b)
    \end{equation*}
    where
    \begin{equation*}
        g(b)
        =
        \frac{\frac{1}{2} b^2 - \qty( 2 \rO + \rL ) b}{\eta (\rO + \rL)^2}
        +
        \frac{2}{\eta}
        .
    \end{equation*}
\end{lemma}

\begin{myproof}
    We show the claim using induction on $\tau$.
    For $\tau = 0$ we want to show
    \begin{equation*}
        0
        \le
        g(b)
        =
        \frac{\frac{1}{2} b^2 - \qty( 2 \rO + \rL ) b}{\eta (\rO + \rL)^2}
        +
        \frac{2}{\eta}
    \end{equation*}
    for all $b$.
    This is true because we can write the r.h.s. as $\frac{(b - \qty( 2 \rO + \rL ) )^2}{2\eta (\rO + \rL)^2} + \frac{\rL (3 \rL+4 \rO)}{2 \eta (\rL+\rO)^2}$ which is clearly positive.
    For $\tau+1$ remaining and $b = \b{T-\tau}{}$ we have
    \begin{align*}
        \sum_{t=T-\tau}^T \x{t}{} (1 - \lambda \b{t}{})
        \le
        \frac{\rO^2 + \qty(\rO^2 + \rL^2)\eta}{(\rO + \rL)^2}
        \tau
        +
        \max\qty{\Big.
            1 - \lambda b
            +
            g\qty\big(b + \eta\rL)
            ,
            g\qty\big(b + \eta(\rL - b))
        }
    \end{align*}
    where the inequality holds by considering the maximizing option of $\x{T-\tau}{} \in \{0, 1\}$ and using the induction hypothesis.
    We will prove that the above is less than $\frac{\rO^2 + \qty(\rO^2 + \rL^2)\eta}{(\rO + \rL)^2}(\tau + 1) + g(b)$, which means that we want to prove that for all $b$
    \begin{align*}
        \frac{\rO^2 + \qty(\rO^2 + \rL^2)\eta}{(\rO + \rL)^2}
        \ge &
        \max\qty{\Big.
            1 - \lambda b
            +
            g\qty\big(b + \eta\rL) - g(b)
            ,
            g\qty\big(b + \eta(\rL - b)) - g(b)
        }
        \\
        = &
        \max\qty{\Big.
            \frac{
                (b-\rL)
                \qty\big(4\rO - (2 - \eta) (b - \rL))
            }{2 (\rL+\rO)^2}
            ,
            \frac{\eta  \rL^2+2 \rO^2}{2 (\rL+\rO)^2}
        }
    \end{align*}
    where for the equality we substituted $\lambda = \frac{\rL}{(\rL + \rO)^2}$ and $g(\cdot)$.
    For the l.h.s. to be bigger than the second term of the maximum, we need $\rO^2 + \qty(\rO^2 + \rL^2)\eta \ge \eta  \rL^2/2 + \rO^2$, or equivalently, $\rO^2 + \rL^2 \ge \rL^2/2$, which is true. 
    For the first term of the maximum, we need
    \begin{equation*}
        2\rO^2 + 2\qty(\rO^2 + \rL^2)\eta
        \ge
        (b-\rL) \qty\big(4\rO - (2 - \eta) (b - \rL))
    \end{equation*}

    The r.h.s. of the above inequality has a global maximum at $\rL + \frac{2 \rO}{2 - \eta}$; substituting this $b$, all we have to prove is that
    \begin{equation*}
        2\rO^2 + 2\qty(\rO^2 + \rL^2)\eta
        \ge
        \frac{4 \rO^2}{2 - \eta}
    \end{equation*}

    We notice 
    \begin{equation*}
        2\rO^2 + 2\qty(\rO^2 + \rL^2)\eta
        \ge
        2\rO^2 + 2\rO^2 \eta
        =
        2\rO^2 (1 + \eta)
        \ge
        2\rO^2 \frac{2}{2-\eta}
    \end{equation*}
    where the last inequality holds because $0 < \eta < 1$.
    This completes the lemma.
\end{myproof}

Now we use the above lemma and bound the Lagrangian.

\begin{corollary}\label{cor:adv:Lagrangian}
    The Lagrangian of \cref{eq:optimization} by $\lambda = \frac{\rL}{(\rL + \rO)^2}$ and $\b{1}{} = \rL$ is for every binary sequence $\{\x{t}{}\}_{t \in [T]}$
    \begin{equation*}
        \lambda T \rO
        +
        \sum_{t=1}^T \x{t}{} \qty(1 - \lambda \b{t}{})
        \le
        \qty(
        \frac{\rO}{\rL+\rO}
        +
        \frac{\rL^2 + \rO^2}{(\rL+\rO)^2} \eta
        )T
        +
        \frac{3 \rL^2+4 \rL \rO+4 \rO^2}{2 (\rL+\rO)^2}
        \frac{1}{\eta}
    \end{equation*}
\end{corollary}

Now we prove \cref{thm:adv:main}.

\begin{myproof}[Proof of \cref{thm:adv:main}]
    By \cref{cor:adv:Lagrangian}, the number of rounds win by Optimizer is at most
    \begin{equation*}
        \qty(
        \frac{\rO}{\rL+\rO}
        +
        \frac{\rL^2 + \rO^2}{(\rL+\rO)^2} \eta
        )T
        +
        \frac{3 \rL^2+4 \rL \rO+4 \rO^2}{2 \eta (\rL+\rO)^2}
        \le
        \qty(
        \frac{\rO}{\rL+\rO}
        +
        \eta
        )T
        +
        \frac{2}{\eta}
    \end{equation*}
    where the inequality follows by $\frac{\rL^2 + \rO^2}{(\rL+\rO)^2} \le 1$ and $\frac{3 \rL^2+4 \rL \rO+4 \rO^2}{2(\rL+\rO)^2} \le \frac{4 \rL^2 + 8 \rL \rO + 4 \rO^2}{2(\rL+\rO)^2} = 2$.
    Setting $\eta = 1/\sqrt T$ bounds the number of wins of the Optimizer by $\frac{\rO}{\rL+\rO} T + 3 \sqrt T$, proving what we want.
\end{myproof}

\section{Omitted Proofs of Section \ref{sec:diff_b}}
\label{sec:app:diff_b}

In this section, we include all the technical proof from \cref{sec:diff_b} that were omitted.

\Convexity*

\begin{myproof}
    The convexity follows from the fact that $\max_i \{ b_i^2 \}$ is a maximum over convex functions and that the rest of the function is affine.

    For the subgradients, we notice that $f$ is piecewise differentiable: whenever $b_i > \max_{j \ne i} b_j$, then the subgradient is unique and equal to $b_i \vec e_i - \vec\rho$.
    For points where multiple $b_i$'s are maximum, we have to take the convex hull of the above gradients for those $i$.
\end{myproof}

\Minimum*

\begin{myproof}
    First, notice that $\vec 0 \in \partial f(\vec 1)$, making $\vec 1$ a global minimum.
    To prove uniqueness, consider an arbitrary $\vec b$ with $\max_i b_i = x$.
    Then $f(\vec b) = \frac{1}{2} x^2 - \vec\rho^\top \vec b + \frac{1}{2} \ge \frac{1}{2} x^2 - x + \frac{1}{2} = (x - 1)^2/2 \ge 0$.
    We notice that the first inequality, since $\rho_i > 0$ for all $i$, can be tight only if $\vec b = x \vec{1}$.
    This means that the only way for $f(\vec b) = 0$ to be true, it has to be that $\vec b = x \vec{1}$ and $x = 1$.
    This proves the proposition.
\end{myproof}

\NormTwoDecrease*

\begin{myproof}
    We have that
    \begin{alignat*}{3}
        \Line{
            \norm{\vecb{t+1}{} - \vec 1}_2^2
        }{=}{
            \norm{\vecb{t}{} - \eta \vec g^{(t)} - \vec 1}_2^2
        }{}
        \\
        \Line{}{=}{
            \norm{\vecb{t}{} - \vec 1}_2^2
            -
            2\eta (\vecb{t}{} - \vec 1)^\top \vec g^{(t)}
            +
            \eta^2 \norm{ \vec g^{(t)} }_2^2
        }{}
        \\
        \Line{}{\le}{
            \norm{\vecb{t}{} - \vec 1}_2^2
            -
            2\eta \qty( f(b_t) - f(\vec 1) )
            +
            \eta^2 \norm{ \vec g^{(t)} }_2^2
        }{\text{convexity}}
    \end{alignat*}

    The proposition follows by noticing that $f(\vec 1) = 0$, $\norm{ \vec g^{(t)} }_2^2 \le 2\norm{\vecb{t}{} - \vec 1}_2^2 + 3$, and that
    \begin{equation*}
        \rhomin\norm{\vecb{}{} - \vec 1}_2^2
        \le
        \sum_i \rho_i \qty(\b{}{i} - 1)^2
        \le
        \bmax^2 - 2 \sum_i \rho_i \b{}{i} + 1
        =
        2 f(\vec b)
    \end{equation*}
\end{myproof}

\ConvergenceSqrtEta*

\begin{myproof}
    We use the fact that $\eta \le \rhomin/10$ in \cref{cl:diff_b:dist_from_one_decrease} to get that for any round $t$
    \begin{equation*}
        \norm{\vecb{t+1}{} - \vec 1}_2^2
        \le
        \qty( 1 - \frac{4}{5}\eta \rhomin)\norm{\vecb{t}{} - \vec 1}_2^2
        +
        3\eta^2
    \end{equation*}
    which inductively proves that for any round $t \ge t_0$
    \begin{equation*}
        \norm{\vec b^{t} - \vec 1}_2^2
        \le
        \qty(1 - \frac{4}{5}\eta \rhomin)^{t - t_0}
        \norm{\vec b^{t_0} - \vec 1}_2^2
        +
        \qty\Bigg( 1 - \qty(1 - \frac{4}{5}\eta \rhomin)^{t - t_0})
        \frac{15 \eta}{4\rhomin}
    \end{equation*}

    Using the fact that $t - t_0 \ge 3 \frac{5}{4\eta\rhomin} \log\frac{1}{\eta}$ and that $\frac{5}{4\eta\rhomin} \ge 5$ we (loosely) get that
    \begin{equation*}
        \qty(1 - \frac{4}{5}\eta \rhomin)^{t - t_0}
        \le
        \eta^3
    \end{equation*}
    that implies
    \begin{equation*}
        \norm{\vec b^{t} - \vec 1}_2^2
        \le
        \eta^3
        \norm{\vec b^{t_0} - \vec 1}_2^2
        +
        \frac{15 \eta}{4\rhomin}
    \end{equation*}
    and proves the lemma.
\end{myproof}

\DistFromAvgDecrease*

\begin{myproof}
    Let $\vec g^{(t)}$ be the update vector: $\vecb{t+1}{} = \vecb{t}{} - \eta \vec g^{(t)}$.
    Then we have
    \begin{alignat*}{3}
        \Line{
            \norm{ \vecb{t+1}{} - \bavg^{(t+1)}\vec 1 }_2^2
        }{=}{
            \norm{ \vecb{t+1}{} - \frac{\vec 1^\top \vecb{t+1}{}}{n} \vec 1 }_2^2
            =\;
            \norm{ \vecb{t}{} - \eta \vec g^{(t)} - \frac{\vec 1^\top \qty(\vecb{t}{} - \eta \vec g^{(t)})}{n}\vec 1 }_2^2
        }{}
        \\
        \Line{}{=}{
            \norm{ \vecb{t}{} - \frac{\vec 1^\top\vecb{t}{}}{n}\vec 1 }_2^2
            -
            2\eta
            \qty( \vec g^{(t)} - \frac{\vec 1^\top\vec g^{(t)}}{n}\vec 1 )^\top
            \qty( \vecb{t}{} - \frac{\vec 1^\top\vecb{t}{}}{n}\vec 1 )
        }{}
        \\
        \Line{}{}{
            \quad+
            \eta^2
            \norm{ \vec g^{(t)} - \frac{\vec 1^\top\vec g^{(t)}}{n}\vec 1 }_2^2
        }{}
        \\
        \Line{}{=}{
            \norm{ \vecb{t}{} - \bavg^{(t)}\vec 1 }_2^2
            -
            2\eta
            \qty( \qty(\vec g^{(t)})^\top\vecb{t}{} - \frac{\qty(\vec 1^\top \vecb{t}{}) \qty(\vec 1^\top \vec g^{(t)})}{n} )
            +
            \eta^2
            \norm{ \vec g^{(t)} - \frac{\vec 1^\top\vec g^{(t)}}{n}\vec 1 }_2^2
        }{}
    \end{alignat*}

    We now analyze the middle term, given that $\vec g^{(t)} = \bmax^{(t)} \vec e_{i^*} - \vec\rho$, where $i^*$ is the winner of that round.
    Dropping the $t$ superscript, we have
    \begin{equation*}
        \vec g^\top \vec b - \frac{\qty(\vec 1^\top \vec b) \qty(\vec 1^\top \vec g)}{n}
        =
        \qty( \bmax \vec e_{i^*} - \vec\rho )^\top \vec b
        -
        \bavg \qty( \bmax - 1 )
        =
        \bmax^2
        -
        \vec\rho^\top \vec b
        -
        \bavg ( \bmax - 1 )
    \end{equation*}

    We now lower bound the above term.
    We set $\vec b = 1 + \vec\d$ and have that the above term is 
    \begin{equation*}
        (1 + \d_{\max})^2 - \vec \rho^\top (1 + \vec\delta) - (1 + \davg) \qty\big(1 + \d_{\max} - 1)
        =
        \d_{\max} - \vec\rho^\top \vec\d
        +
        \d_{\max}^2 - \d_{\max} \davg
    \end{equation*}

    We then notice that
    \begin{align*}
        \d_{\max} - \rho^\top \d
        &=
        \sum_i \rho_i \qty( \d_{\max} - \d_i )
        \ge
        \rhomin \sum_i \qty( \d_{\max} - \d_i )
        =
        \rhomin \norm{\d_{\max}\vec 1 - \vec\d}_1
        \\&\ge
        \rhomin \norm{\d_{\max}\vec 1 - \vec\d}_2
        \ge
        \rhomin \norm{\davg\,\vec 1 - \vec\d}_2
        =
        \rhomin \norm{\vec b - \bavg\vec 1}_2
    \end{align*}
    and that
    \begin{equation*}
        \d_{\max}^2 - \d_{\max} \davg
        \ge
        - 2 D^2 \eta
    \end{equation*}
    which follows from our assumption that $\norm*{\vec b - \vec 1}_2 = \norm*{\vec\d}_2 \le D \sqrt\eta$.
    Finally, we bound
    \begin{equation*}
        \norm{ \vec g - \frac{\vec 1^\top\vec g}{n}\vec 1 }_2^2
        =
        \bmax^2 + \norm{\vec\rho}_2^2 + n (\bmax - 1)^2 - 2\bmax \rho_{i^*} - \frac{\bmax(\bmax-1)}{n} + \frac{\bmax-1}{n}
        \le
        3
    \end{equation*}
    where the inequality can be proven from $\norm{\vec b - \vec 1}_2 \le D \sqrt \eta$ and $\eta\le\frac{1}{2 n D^2}$.
\end{myproof}

\ConvergenceDistFromAvg*

\begin{myproof}
    First, let $t_1$ be such that for all rounds $t \in [t_0, t_1)$ it holds $\norm*{\vecb{t}{} - \bavg^{(t)}}_2 \ge 3 \frac{D^2 + 1}{\rhomin} \eta$.
    Then using \cref{cl:diff_b:dist_from_avg_decrease} and summing the inequality over all $t \in [t_0, t_1)$ we have that
    \begin{alignat*}{3}
        \Line{
            \norm{\vecb{t_1}{} - \bavg^{(t_1)} \vec 1 }_2^2
        }{\le}{
            \norm{\vecb{t_0}{} - \bavg^{(t_0)} \vec 1}_2^2
            -
            \sum_{t \in [t_0, t_1)} \qty(
                2 \rhomin \eta \norm{\vecb{t}{} - \bavg^{t} \vec 1}_2
                -
                4 \eta^2 \qty( D^2 + 1 )
            )
        }{}
        \\
        \Line{}{\le}{
            D^2 \eta
            -
            \sum_{t \in [t_0, t_1)} \qty(
                2 \rhomin \eta 3 \frac{D^2 + 1}{\rhomin} \eta
                -
                4 \eta^2 \qty( D^2 + 1 )
            )
        }{}
        \\
        \Line{}{=}{
            D^2 \eta
            -
            (t_1 - t_0) 2 \eta^2 (D^2 + 1)
            \le
            D^2 \eta \qty\big( 1 - 2 (t_1 - t_0) \eta )
        }{}
    \end{alignat*}
    Since the l.h.s. is non-negative, it must hold that $t_1 - t_0 \le \frac{1}{2 \eta}$.
    This proves that there exists a $t_1 > t_0 + \frac{1}{2 \eta}$ such that
    \begin{equation*}
        \norm{\vecb{t_1}{} - \bavg^{(t_1)}\vec 1}_2
        <
        3 \frac{D^2 + 1}{\rhomin} \eta
    \end{equation*}

    We now inductively show that this holds for all $t \ge t_1$.
    Fix a $t$ where this holds and notice that because of \cref{cl:diff_b:dist_from_avg_decrease}
    \begin{equation*}
        \norm{\vecb{t+1}{} - \bavg^{(t+1)} \vec 1}_2^2
        \le
        \norm{\vecb{t}{} - \bavg^{(t)}\vec 1}_2^2
        -
        2 \rhomin \eta \norm{\vecb{t}{} - \bavg^{(t)}\vec 1}_2
        +
        4 \eta^2 \qty( D^2 + 1 )
    \end{equation*}

    Considering the above inequality, we examine two cases:
    \begin{itemize}
        \item If $\norm*{\vecb{t}{} - \bavg^{(t)} \vec 1}_2 \ge 2 \frac{D^2 + 1}{\rhomin} \eta$, then the inequality shows $\norm*{\vecb{t+1}{} - \bavg^{(t+1)} \vec 1}_2^2 \le \norm*{\vec b^{t} - \bavg^{t} \vec 1}_2^2$.

        \item If $\norm*{\vecb{t}{} - \bavg^{(t)} \vec 1}_2 < 2 \frac{D^2 + 1}{\rhomin} \eta$ then
        let $R(y) = y^2 - 2 \rhomin \eta y + 4 \eta^2 \qty(D^2+1)$
        denote the right-hand side as a function of $y = \norm*{\vecb{t}{} - \bavg^{(t)} \vec 1}_2$.
        We will show that $R(y) < \qty\big(3 \frac{D^2 + 1}{\rhomin} \eta)^2$.
        Observe that $R(y)$ is a convex function of $y$, so its maximum value on the
        interval $\qty\big[0,2 \frac{D^2 + 1}{\rhomin} \eta]$ is attained at one of the endpoints
        of that interval. At $y=0$ we have 
        \[
          R(y) = 4 \eta^2 \qty( D^2 + 1 ) 
            = \qty(3 \frac{D^2 + 1}{\rhomin} \eta)^2 \cdot \qty( \frac{4}{9} \cdot \frac{\rhomin^2}{D^2 + 1} )
                <
                \qty(3 \frac{D^2 + 1}{\rhomin} \eta)^2 .
        \]
        At $y=2 \frac{D^2 + 1}{\rhomin} \eta \leq 1$ we have
        \[
            R(y) = y^2 - (2 \rhomin \eta) \qty(2 \frac{D^2+1}{\rhomin} \eta) + 4 \eta^2 \qty(D^2+1)
                = y^2 =
                \qty(2 \frac{D^2 + 1}{\rhomin} \eta)^2
                <
                \qty(3 \frac{D^2 + 1}{\rhomin} \eta)^2 .
        \]
    \end{itemize}

    The above case analysis proves the lemma.
\end{myproof}

\ConvergenceDistFromOneEta*

\begin{myproof}
    Fix a $t \ge t_0$.
    We know that $\bavg^{(t+1)} = \bavg^{(t)} + \frac{\eta}{n}(1 - \bmax^{(t)})$ by adding the update rules of all the players.
    We also know that $\bmax^{(t)} \ge \bavg^{(t)}$ and by the assumptions that $\bmax^{(t)} \le \bavg^{(t)} + A \eta$; this implies that
    \begin{equation*}
        \bavg^{(t)} \qty(1 - \frac{\eta}{n}) + \frac{\eta}{n}\qty(1 - A \eta)
        \le
        \bavg^{(t+1)}
        \le
        \bavg^{(t)} \qty(1 - \frac{\eta}{n}) + \frac{\eta}{n}
    \end{equation*}

    Applying this inductively, we get that for any $t \ge t_0$
    \begin{equation*}
        \bavg^{(t_0)} \qty(1 - \frac{\eta}{n})^{t-t_0}
        +
        \qty(1 - \qty(1 -\frac{\eta}{n})^{t-t_0})\qty(1 - A \eta)
        \le
        \bavg^{(t)}
        \le
        \bavg^{(t_0)} \qty(1 - \frac{\eta}{n})^{t-t_0}
        +
        \qty(1 - \qty(1 -\frac{\eta}{n})^{t-t_0})
        .
    \end{equation*}

    Taking $t - t_0 \ge 3 \frac{n}{\eta} \log\frac{1}{\eta}$ and using the fact that $n/\eta \ge 4$, we have that $\qty(1 -\frac{\eta}{n})^{t-t_0} \le \eta^3$ which implies
    \begin{equation*}
        \qty(1 - \eta^3)\qty(1 - A \eta)
        \le
        \bavg^{(t)}
        \le
        \bavg^{(t_0)} \eta^3
        +
        1
        .
    \end{equation*}

    We now get the following by noticing that $\bavg^{(t_0)} \le 1 + D \sqrt \eta \le $ and $(1- \eta^3)(1-A\eta) \ge 1 - A\eta -\eta^2$ (which follows from $D^2 \le \frac{1}{2\eta}$):
    \begin{equation*}
        1
        - A \eta - \eta^2
        \le
        \bavg^{(t)}
        \le
        1
        +
        2 \eta^3
    \end{equation*}

    We get the final result by noticing that by the lemma's assumptions, $\abs{\b{t}{i} - \bavg^{(t)}} \le A\eta$.
\end{myproof}

\FinalDistFromOne*

\begin{proof}
    By \cref{lem:diff_b:sqrt_eta}, for rounds $t \ge t_1 := 1 + \frac{15}{4\eta\rhomin} \log\frac{1}{\eta}$ we have that
    \begin{equation*}
        \norm{\vec b^{t+\tau} - \vec 1}_2^2
        \le
        \frac{15 \eta}{4\rhomin}
        +
        \eta^3
        n \max_i \qty(\b{1}{i} - 1)^2
        \le
        \frac{15 \eta}{4\rhomin}
        +
        \eta^2
    \end{equation*}

    By \cref{lem:diff_b:dist_from_avg_eta} (with $D^2 \gets \frac{15}{4\rhomin} + \eta$), we get that for $t \ge t_1 + t_2$, where $t_2 = 1 + \frac{1}{2\eta}$ we have that
    \begin{equation*}
        \norm{\vecb{t}{} - \bavg^{(t)}\vec 1}_2
        \le
        3 \frac{D^2 + 1}{\rhomin} \eta
        \le
        \frac{129}{10\rhomin^2} \eta
    \end{equation*}

    We now use \cref{lem:diff_b:eta_dist_from_one} (with $D^2 \gets \frac{15}{4\rhomin} + \eta$) to get that $t \ge t_1 + t_2 + t_3$, where $t_3 = 3\frac{n}{\eta}\log\frac{1}{\eta}$ to get that for all $i$
    \begin{equation*}
         1 - 2A \eta - \eta^2
         \le \b{t}{i} \le
         1 + A \eta + 2\eta^3
    \end{equation*}
    which implies 
    \begin{equation*}
         1 - \frac{129}{5\rhomin^2} \eta - \eta^2
         \le \b{t}{i} \le
         1 + \frac{129}{10\rhomin^2} \eta + 2\eta^3
    \end{equation*}

    While this proves that the bids are $\order*{\eta}$ close to $1$, we can get a better bound by re-applying the lemmas.
    Specifically, we can use \cref{lem:diff_b:dist_from_avg_eta} with $D^2 \gets \frac{667}{\rhomin^2}\eta$ to get that for $t \ge t_1 + t_2 + t_3 + t_4$ with $t_4 = 1 + \frac{2}{\eta}$
    \begin{equation*}
        \norm{\vecb{t}{} - \bavg^{(t)}\vec 1}_2
        \le
        3 \frac{\frac{667}{\rhomin^2}\eta + 1}{\rhomin} \eta
        \le
        \frac{6}{\rhomin} \eta
    \end{equation*}
    and applying \cref{lem:diff_b:eta_dist_from_one} we get that for $t \ge t_1 + t_2 + t_3 + t_4 + t_5$ with $t_5 = 3\frac{n}{\eta}\log\frac{1}{\eta}$ for all $i \in [n]$
    \begin{equation*}
         1 - \frac{12}{\rhomin} \eta - \eta^2
         \le \b{t}{i} \le
         1 + \frac{6}{\rhomin} \eta + 2\eta^3
    \end{equation*}
    which proves the theorem.
\end{proof}

\DiscrepancyDiffBudgets*

\begin{myproof}[Proof of \cref{thm:diff_b:main}]
    Fix an agent $i$, an interval $[t_1, t_2]$ of length $\tau$, and assume that $t_1 \ge t \ge \frac{11}{\eta \rhomin} \log\frac{1}{\eta}$.
    For any $t \in [t_1, t_2]$, by \cref{thm:diff_b:bids_are_one_pm_eta} it holds that
    \begin{equation*}
         1 - \frac{12}{\rhomin} \eta - \eta^2
         \le \b{t}{i} \le
         1 + \frac{6}{\rhomin} \eta + 2\eta^3
    \end{equation*}
    which lets us use \cref{lem:diff_b:wins_deviation} with $m = 1 - \frac{12}{\rhomin} \eta - \eta^2$ and $M = 1 + \frac{6}{\rhomin} \eta + 2\eta^3$.
    This proves that the number of times agent $i$ wins is
    \begin{equation*}
        \sum_{t=t_1}^{t_2-1} \x{t}{i}
        \ge
        \rho_i \tau
        -
        \rho_i \tau \qty(1 - \frac{1}{M})
        -
        \frac{M - m}{\eta M}
        \ge
        \rho_i \tau 
        -
        \rho_i \tau \frac{6}{\rhomin} \eta
        -
        \frac{18}{\rhomin}
        .
    \end{equation*}

    Substituting $\eta = 1/\sqrt T$ proves the theorem.
\end{myproof}

\end{document}